\documentclass[preprint,12pt]{elsarticle}

\usepackage{amssymb}
\usepackage{amsthm}
\usepackage{mathrsfs}
\usepackage{graphicx}
\usepackage{subcaption}
\usepackage{amssymb}
\usepackage{lineno}
\usepackage{hyperref}
\usepackage{amsmath}
\usepackage{multicol}
\usepackage{multirow}
\usepackage{pifont}
\usepackage{natbib}
\usepackage{fleqn}
\usepackage{txfonts}
\usepackage{csquotes}
\usepackage{color}

\DeclareMathOperator{\supp}{supp}

\newtheorem{theorem}{Theorem}

\newtheorem{proposition}{Proposition}
\newtheorem{corollary}{Corollary}

\newdefinition{definition}{Definition}
\newdefinition{example}{Example}
\newdefinition{remark}{Remark}

\journal{Mathematical Biosciences}

\begin{document}

\begin{frontmatter}

\title{Species subsets and embedded networks of S-systems}

\author[mmsu,dlsu]{Honeylou F. Farinas}
\ead{farinas\_ honey@yahoo.com}
\author[dlsu,imsp,max,lmu]{Eduardo R. Mendoza}
\ead{mendoza@lmu.de}
\author[dlsu]{Angelyn R. Lao \corref{cor1}}
\ead{angelyn.lao@dlsu.edu.ph}

\cortext[cor1]{Corresponding author. Phone number: (632) 536 0270. Present address: Mathematics and Statistics Department, De La Salle University, Manila, 0922 Philippines}

\address[mmsu]{Department of Mathematics, Mariano Marcos State University, Ilocos Norte, 2906 Philippines}
\address[dlsu]{ Mathematics and Statistics Department, De La Salle University, Manila, 0922 Philippines}
\address [imsp]{Institute of Mathematical Sciences and Physics, University of the Philippines Los Ba\~{n}os, Laguna, 4031 Philippines}
\address[max]{Max Planck Institute of Biochemistry, Martinsried, Munich, Germany}
\address[lmu] {Faculty of Physics, Ludwig Maximilians University, Geschwister -Scholl- Platz 1, 80539 Munich Germany}

\begin{abstract}
Magombedze and Mulder (2013) studied the gene regulatory system of  \textit{Mycobacterium Tuberculosis} (\textit{Mtb}) by partitioning this into three subsystems based on putative gene function and role in dormancy/latency development. Each subsystem, in the form of $S$-system, is represented by an embedded chemical reaction network (CRN), defined by a species subset and a reaction subset induced by the set of digraph vertices of the subsystem. Based on the network decomposition theory initiated by Feinberg in 1987, we have introduced the concept of incidence-independent and developed the theory of $\mathscr{C}$- and $\mathscr{C}^*$-decompositions including their structure theorems in terms of linkage classes. With the $S$-system CRN $\mathscr{N}$ of Magombedze and Mulder's \textit{Mtb} model, its reaction set partition induced decomposition of subnetworks that are not CRNs of $S$-system but constitute independent decomposition of $\mathscr{N}$. We have also constructed a new $S$-system CRN $\mathscr{N}^*$ for which the embedded networks are $\mathscr{C}^*$-decomposition. We have shown that subnetworks of $\mathscr{N}$ and the embedded networks (subnetworks of $\mathscr{N}^*$) are digraph homomorphisms. Lastly, we attempted to explore modularity in the context of CRN.
\end{abstract}


\end{frontmatter}


\section{Introduction}
\label{intro}
$S$-system models are typically derived from biochemical maps \cite{VOIT2000}, which are based on digraphs whose vertices and arcs represent molecules and interactions of a biochemical system. Modularity of the system is often expressed by partitioning the vertices into subsets and forming corresponding subsystems. In the embedded chemical reaction network (CRN) representation of the S-system \cite{AJLM2017}, each subsystem is represented by an embedded network, defined by a species subset and a reactions subset induced by the set of digraph vertices of the subsystem. In this paper, we study the relationship between the $S$-system CRN and its embedded networks to gain insight into some properties of the modelled biochemical system.

The paper was motivated by our analysis of two $S$-system models of the gene regulatory system of \textit{Mycobacterium Tuberculosis} (\textit{Mtb}) by Magombedze and Mulder \cite{MAMU2013}. In those models, the vertices (representing genes) were subdivided into three subsets based on putative gene function. We use the $S$-systems and their subsystems as running examples in the paper.

Our approach is based on the decomposition theory of chemical reaction networks, which was initiated by M. Feinberg in his 1987 review \cite{FEIN1987}. There, he also introduced the important concept of an independent decomposition: a
 decomposition \{$\mathscr{N}_1,...,\mathscr{N}_k$\} of $\mathscr{N}$ is independent if the stoichiometric subspace of $\mathscr{N}$ is the direct sum of the stoichiometric subspaces of the subnetworks $\mathscr{N}_i$. In \cite{FLRM2017}, it was shown that the deficiency of the network is bounded by the sum of the deficiencies of the subnetworks of an independent decomposition, i.e. $\delta \leq \delta_1 + ...+ \delta_k$.
 
 We introduce the complementary concept of incidence-independence, i.e. when the image of the network's incidence map is the direct sum of the images of the subnetworks' incidence maps. The primary example of an incidence-independent decomposition are the linkage classes of a network. For incidence-independent decompositions, $\delta \geq \delta_1 + ...+ \delta_k$ , a property familiar from linkage classes. A bi-independent decomposition is one which is both independent and incidence-independent, and for which $\delta = \delta_1 + ...+ \delta_k$.
 
 We observed that the embedded networks of the \textit{Mtb} $S$-systems had the distinctive property that the zero complex was the only complex common to them. We abstracted the concept of a $\mathscr{C^*}$-decompositions from this: in such a decomposition, the non-zero complexes partition the set of the  non-zero complexes of the network. An interesting subset of $\mathscr{C^*}$-decompositions are the $\mathscr{C}$-decompositions, where the partition of the reaction set (defining the decomposition) also partitions the set of complexes. Again, the primary example of a $\mathscr{C}$-decomposition are the linkage classes. We derive structure theorems for $\mathscr{C}$-decompositions and $\mathscr{C^*}$-decompositions in terms of linkage classes.

 One of our main results is a framework consisting of three components:
 \begin{itemize}
 \item $\mathscr{N}$ and an independent decomposition into $S$-system similar subnetworks
 \item a new $S$-system CRN $\mathscr{N}^*$ for which the embedded networks are a $\mathscr{C}^*$-decomposition
 \item digraph homomorphisms from subnetworks of $\mathscr{N}$ to the corresponding embedded networks ($=$ subnetworks $\mathscr{N}^*$)
 \end{itemize} 
 
  We have organized the contents of the paper as follows:
 
 Section 2 collects the concepts and known results from the theory of chemical reaction networks underlying our results.
 
 Section 3 introduces the NRP (Non-Replicating Phase) and STR (Stationary Phase) models of \textit{Mtb's} gene regulatory system as well as the subsystems of genes $\mathcal{D}_1,\mathcal{D}_2$ and $\mathcal{D}_3$.
 
 Section 4 presents a correction of Proposition 10 of \cite{AJMM2015}, which was motivated by the calculation of the deficiency of the NRP CRN, whose value was 39, instead of the predicted 40. The correct formula is an inequality and a new proof using decomposition theory is provided. Incidence-independence of the species decomposition of the $S$-system CRN is a sufficient condition for equality, as in the original proposition.
 
 Section 5 collects and analyzes the properties of the NRP and STR networks as well as those of their embedded networks. A comparison with those from the BST Case studies of \cite{AJLM2017} show a very high degree of coincidence. Concordance tests of current tools cannot handle the large networks we have been studying. However for $S$-system CRNs, we have found a proof of discordance for all networks with at least two species.
 
 Section 6 develops the theory of $\mathscr{C}$ - and $\mathscr{C}^*$ -decompositions including their structure theorems in terms of linkage classes. We then construct $\mathscr{N}^*$, the $*$-disjoint union of embedded networks and show that the embedded networks constitute an independent $\mathscr{C}^*$ -decomposition of $\mathscr{N}^*$. 
 
 In the final results section (Section 7), we complete the previously mentioned framework by introducing the $S$-system similar subnetworks $\mathscr{N}_i$ that constitute an independent decomposition of $\mathscr{N}$ as well as the digraph homomorphisms $\mathscr{N}_i \rightarrow \mathscr{N}_i^*$. We then show that the modularity of the decomposition of the $S$-system $\mathscr{N}$ (considered as a digraph division) is not consistent with the modularity of the original digraph model of the gene regulatory system. We introduce the concept of species coupling level to suggest that for a CRN, stoichiometric properties are needed for an appropriate modularity concept.
 

\section{Fundamentals of chemical reaction networks}
\label{sec:2}
In this section, we review concepts of subnetworks and embedded networks, decomposition and CRN representation of an $S$-system. For a background on Chemical Reaction Network Theory (CRNT), the reader is referred to \ref{app:fundametals of CRNT}. The discussion
includes the notations, concepts and results pertinent to this paper.

\subsection{Subnetworks and embedded networks of a CRN}
A subnetwork and an embedded network of a CRN were introduced by Joshi and Shui in \cite{JOSH2013} and \cite{JOSH2012}, respectively. The definitions are based on the concept of restriction maps between subsets in a network's sets.

\begin{definition}
\label{def:restriction}
Let $(\mathscr{S}, \mathscr{C}, \mathscr{R})$ be a chemical reaction network.
Consider a subset of the species $S \subset \mathscr{S}$, a subset of the complexes $C \subset \mathscr{C}$ and a subset of the reactions $R \subset \mathscr{R}$.
\begin{enumerate}
\item The \textbf{restriction of $R$ to $S$}, denoted by $R |_S$, is the set of reactions obtained by taking the reactions in $R$ and removing all species not in $S$ from the reactant and product complexes. If a reactant or a product complex does not contain any species from the set $S$, then the complex is replaced by the $0$ complex in $R |_S$. If a trivial reaction (one in which the reactant and product complexes are the same) is obtained in this process, then that reaction is removed. Also, extra copies of repeated reactions are removed.
\item The \textbf{restriction of $C$ to $R$}, denoted by $C |_R$, is the set of (reactant and product) complexes of the reactions in $R$.
\item The \textbf{restriction of $S$ to $C$}, denoted by $S |_C$, is the set of species that are in the complexes in $C$.
\end{enumerate}
\end{definition}

\begin{definition}
Let $(\mathscr{S}, \mathscr{C}, \mathscr{R})$ be a chemical reaction network. 
\begin{enumerate}
\item A subset of the reactions $\mathscr{R}' \subset \mathscr{R}$ defines the \textbf{subnetwork} $(\mathscr{S} |_{\mathscr{C} |\mathscr{R}'}, \mathscr{C} | _{\mathscr{R}'}, \mathscr{R}')$, where $\mathscr{C}|_{\mathscr{R}'}$ denotes the set of complexes that appear in the reactions $\mathscr{R}'$ and $\mathscr{S} |_{\mathscr{C} | \mathscr{R}'}$ denotes the set of species that appear in those complexes.
\item An \textbf{embedded network} of $(\mathscr{S}, \mathscr{C}, \mathscr{R})$, which is defined by a subset of the species, $S=\{X_{i1}, X_{i2}, \cdots, X_{ik} \}\subset \mathscr{S}$, and a subset of the reactions, $R=\{R_{j1}, R_{j2}, \cdots, R_{jl} \}\subset \mathscr{R}$, that involve all species of $S$, is the network $\left( S, \mathscr{C}|_{R|_{S}}, R|_S \right)$ consisting of the reactions $R|_{S}$.
\end{enumerate}
\end{definition}

In the context of this paper, the subnetworks and embedded networks of the CRN representation of an $S$-system are not necessarily the same. We used $\mathscr{N}_i$ and $\mathscr{N}_i^*$ to denote the subnetwork and embedded network of a CRN $\mathscr{N}$, respectively. The set of subnetworks $\mathscr{N}_i$ is obtained by removing a subset of reactions, while an embedded network $\mathscr{N}_i^*$ is obtained by removing a subset of reactions or subsets of species or both. For instance, removing the species $X_2$ from the reaction $\displaystyle{X_1 + X_2 \rightarrow X_1 + X_3}$ results in the reaction
$X_1 \rightarrow X_1 + X_3$.

In some cases, removing species results in a trivial reaction – where the source and product complex are identical. For instance, removal of both $X_2$ and $X_3$ from $X_1 + X_2 \rightarrow X_1 + X_3$ results in the trivial reaction $X_1 \rightarrow X_1$. So, after removing species, any trivial reactions and any copies of duplicate reactions are discarded.

\begin{example}
Consider the $S$-system ARL3-S of Arceo et. al in \cite{AJLM2017} with 4 species $\mathscr{S}=\{X1, X2, X3, X4\}$ and 8 reactions $\mathscr{R}=\{R1, R2,...,R8\}$ with the following CRN $\mathscr{N}$:
\begin{align*}
R1:&X1 \rightarrow 2X1   		& R5:&X2+X4 \rightarrow X3+X2+X4\\              
R2:&X1+ X2+X4 \rightarrow X2+X4 & R6:&X3 \rightarrow ) \\
R3:&X1+X2 \rightarrow X1+2X2	& R7:&X2+X3 \rightarrow X3+X2+X4\\
R4:&X1+X2 \rightarrow X1		& R8:&X4 \rightarrow 0
\end{align*}
Suppose that we will form 2 subnetworks $(\mathscr{N}_1$ and $\mathscr{N}_2)$ and 2 embedded networks  $(\mathscr{N}_1^*$ and $\mathscr{N}_2^*)$ of $\mathscr{N}$. To form the subnetworks $\mathscr{N}_1$ and $\mathscr{N}_2$ , we simply partition the eight reactions. Suppose we partition the reactions as follows:  $\mathscr{R}_{\mathscr{N}_1}=\{R1, R2, R3, R4\}$ and $\mathscr{R}_{\mathscr{N}_2}=\{R5, R6, R7, R8\}$. Thus, $\mathscr{N}_1$ and $\mathscr{N}_2$ consist of the set of reactions $\mathscr{R}_{\mathscr{N}_1}$ and $\mathscr{R}_{\mathscr{N}_2}$, respectively. Now, for the embedded networks $\mathscr{N}_1^*$ and $\mathscr{N}_2^*$, we partition the four species. Suppose we have the following partitions: $\mathscr{S}_{\mathscr{N}_1^*}=\{X1, X2\}$ and $\mathscr{S}_{\mathscr{N}_2^*}=\{X3, X4\}$. Applying (1) of Definition \ref{def:restriction}, we have the following set of reactions: $\mathscr{R}_{\mathscr{N}_1^*}= \{ X1 \rightarrow 2X1, X1 + X2 \rightarrow X2, X1+X2 \rightarrow X1+2X2, X1 + X2 \rightarrow X1\}$ and $\mathscr{R}_{\mathscr{N}_2^*}
= \{ 0 \rightarrow X3, X3 \rightarrow 0, X3 \rightarrow X3 +X4, X4 \rightarrow 0 \} $
\end{example}

\subsection{Decomposition of a CRN}
\label{decomposition}
Decomposition Theory was initiated by Feinberg in his 1987 review paper. In this paper, we introduced several decompositions as a result of our CRN analysis of Magombedze and Mulder's \cite{MAMU2013} \textit{Mtb} $S$-systems. The incidence-independent and species decompositions will be discussed in Section \ref{sec:2} and the $\mathscr{C}^*$ and $\mathscr{C}$-decompositions in Section \ref{sec:6.1}. Their relationship with that of Feinberg's decomposition are shown in \ref{app:network decomposition}.
 
Feinberg's general concept of a network decomposition is as follows:
\begin{definition}
\label{def:subnetwork}
A set of subnetworks $\mathscr{N}_i=(\mathscr{S}_i, \mathscr{C}_i, \mathscr{R}_i)$ is a \textbf{network decomposition} of $\mathscr{N}=(\mathscr{S},\mathscr{C}, \mathscr{R})$ if $\{\mathscr{R}_i\}$ forms a partition of $\mathscr{R}$.
\end{definition}

 
 Note that partitioning the reaction set does not necessarily partition the set of complexes. A basic property of decomposition is that ${s} \leq \sum s_i$ where $s$ and $s_i$ are the dimensions of the stoichiometric subspaces of the network $\mathscr{N}$ and its corresponding subnetworks $\mathscr{N}_i$, respectively. Meanwhile, if they are equal, then the network is  \textbf{independent} \cite{FEIN1987}.
 
 The best known and most studied decomposition is the linkage class decomposition. Independence of linkage classes (ILC) is an important feature of a CRN. Boros \cite{BOROS2013} showed that a necessary and sufficient condition for a  linkage class to be independent is $\delta=\sum \delta_i$.  Generally, if a network has an independent decomposition then $\delta \leq \sum \delta_i$ \cite{FLRM2017}.

When considering partitions of the reaction set of a CRN, we also call the decompositions a refinement or coarsening of the decomposition. 
 
\begin{definition}
 If $\mathcal{P}=\{\mathcal{P}_i\}$ and $\mathcal{P'}=\{\mathcal{P'}_j\}$ are partitions of a set, then $\mathcal{P}$ is a \textbf{refinement} of $\mathcal{P'}$ if each $\mathcal{P}_i$ is contained in (exactly) one $\mathcal{P'}_j$. 
\end{definition}
 
It is easy to show that this property is equivalent to each $\mathcal{P'}$ being the disjoint union of some $\mathcal{P}_i$'s. We also say the $\mathcal{P}$ is \textbf{finer} than $\mathcal{P'}$, $\mathcal{P'}$ is \textbf{coarser} than $\mathcal{P}$ and $\mathcal{P'}$ is a \textbf{coarsening} of $\mathcal{P}$.
 
We have the following elementary but useful Proposition:
 
\begin{proposition}
\label{prop: independent coarsening}
If a decomposition is independent, then any coarsening of the decomposition is independent.
\end{proposition} 
 
\begin{proof}
 Suppose $x$ is in the intersection of the stoichiometric subspaces of two subnetworks of a coarsening. Since each stoichiometric subspace is the direct sum of subspaces from the independent refinement, then the $x$ is the sum of elements from each subnetwork. It follows that $x=0$.
\end{proof}

 \subsection{CRN representation of an S-system}
 
 An $S$-system is a set of first-order differential equations that all have the same structure: the derivative of a variable is equal to the difference of two products of power-law functions. The formalism is based on using linear Taylor approximation in logarithmic space. An $S$-system with $m$ dependent variables $X_i$ is defined as
$$\dfrac{dX_i}{dt}=\alpha_i\displaystyle\prod_{j=1}^{m}X_j^{g_{ij}}-
\beta_i\displaystyle\prod_{j=1}^{m}X_j^{h_{ij}},\; i\in\{1,\dots,m\}$$
where the dependent variables and the constants $\alpha_i$ and $\beta_i$ are nonnegative, and the exponents $g_{ij}$ and $h_{ij}$ are real. 

To enable the use of tools from CRNT in enhancing understanding of $S$-system dynamics, Arceo et. al  introduced in \cite{AJMM2015} the CRN representations of an $S$-system using a biochemical map. A biochemical map for such an $S$-system is the following:

\begin{figure}[h!]
\begin{center}
    \includegraphics[width=0.5\textwidth]{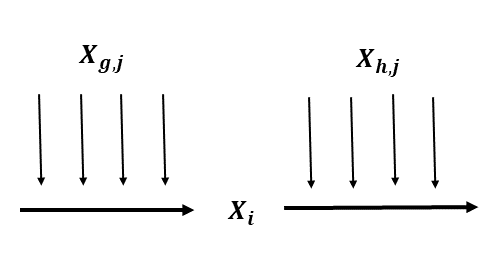}
\end{center}
\caption{Biochemical map of an $S$-system.}
\label{fig:S system}
\end{figure}

\noindent
For each dependent variable $X_i$, where $X_{g,j}$ is a variable with $g_{ij} \neq 0$ in the production term and $X_{h,j}$ is one with $h_{ij} \neq 0$, $i \neq j$ in the degradation term. The CRN representation of this biochemical map is the CRN representation of the $S$-system. In other words, for each $X_i$, the input reaction is $\sum X_{g,j} \rightarrow X_i + \sum X_{g,j}$ and the output reaction is $X_i + \sum X_{h,j} \rightarrow \sum X_{h,j}$. For each independent variable $X_k$, we add a reaction $0 \rightarrow X_k$. (see Figure \ref{fig:S system})


\section{The running examples: S-system models of gene regulatory systems of \textit{Mycobacterium tuberculosis}}
\label{sec:3}

Magombedze and Mulder \cite{MAMU2013} developed two models of gene regulatory systems of \textit{Mycobacterium tuberculosis} that used principles of biochemical pathway modelling of systems biology into S-systems. The models correspond to two phases of the pathogen's life cycle: the Non-Replicating Phase (NRP) and  the Stationary Phase (STR). The NRP model simulates the development of dormant \textit{Mtb} by gradual oxygen depletion over 80 days. While the STR model simulates latent \textit{Mtb} over 60 days. The gene regulatory system of \textit{Mtb} was formalized as digraphs (see Figure ~\ref{fig:NRPSTR}), with the vertices representing the genes and the arcs representing the interactions between them. 

An S-system - with the digraph as its \enquote{biochemical map}- is then built for both NRP (left side of Figure \ref{fig:NRPSTR}) and STR (right side of Figure \ref{fig:NRPSTR}) models with the genes (vertices of the digraph) as variables. The models assumed that gene interactions will result in the change of expression of one gene, in such a way that all interactions should contribute to the term which is modelled as a function responsible for gene regulation in the system. Thus, the $S$-systems for NRP and STR considered the genes as the dependent variables $X_i$ and no independent variables. It is represented as
$$\frac{dX_i}{dt}=\alpha_i \prod_{j=1}^{m} X^{g_{ij}}_j - \beta_i X_i, \quad i\in {1,...,m} $$
where the constants $\alpha_i$ and $\beta_i$ are nonnegative and the exponent $g_{i,j}$ is real. The S-system ODEs for the NRP and STR models can be found in the Supplementary Information of \cite{MAMU2013}.

\begin{figure}
\begin{center}
\begin{minipage}[b]{0.45\textwidth}
    \includegraphics[width=1\textwidth]{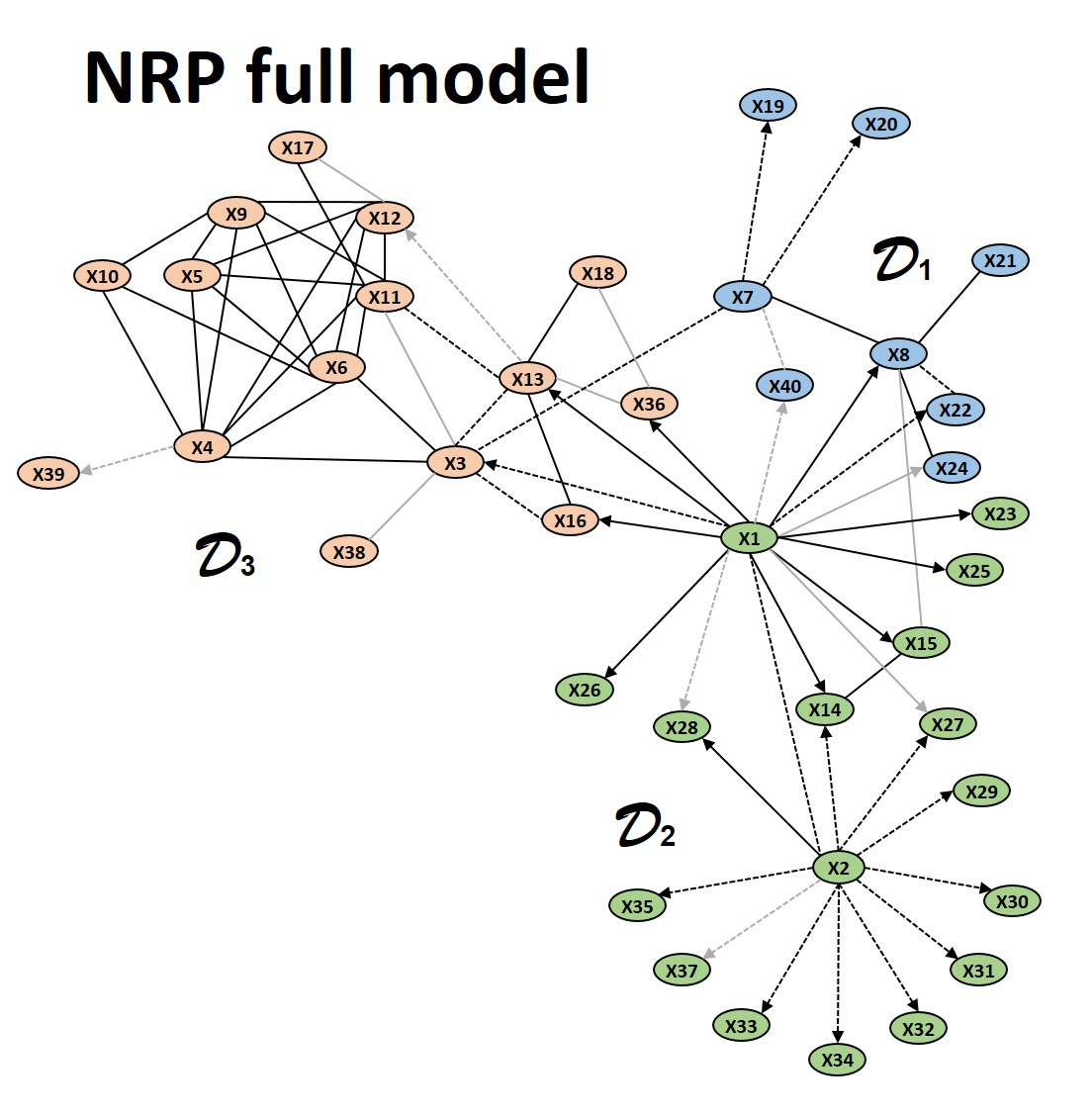}
  \end{minipage}
  \hspace{0.5cm}
  \begin{minipage}[b]{0.45\textwidth}
    \includegraphics[width=1\textwidth]{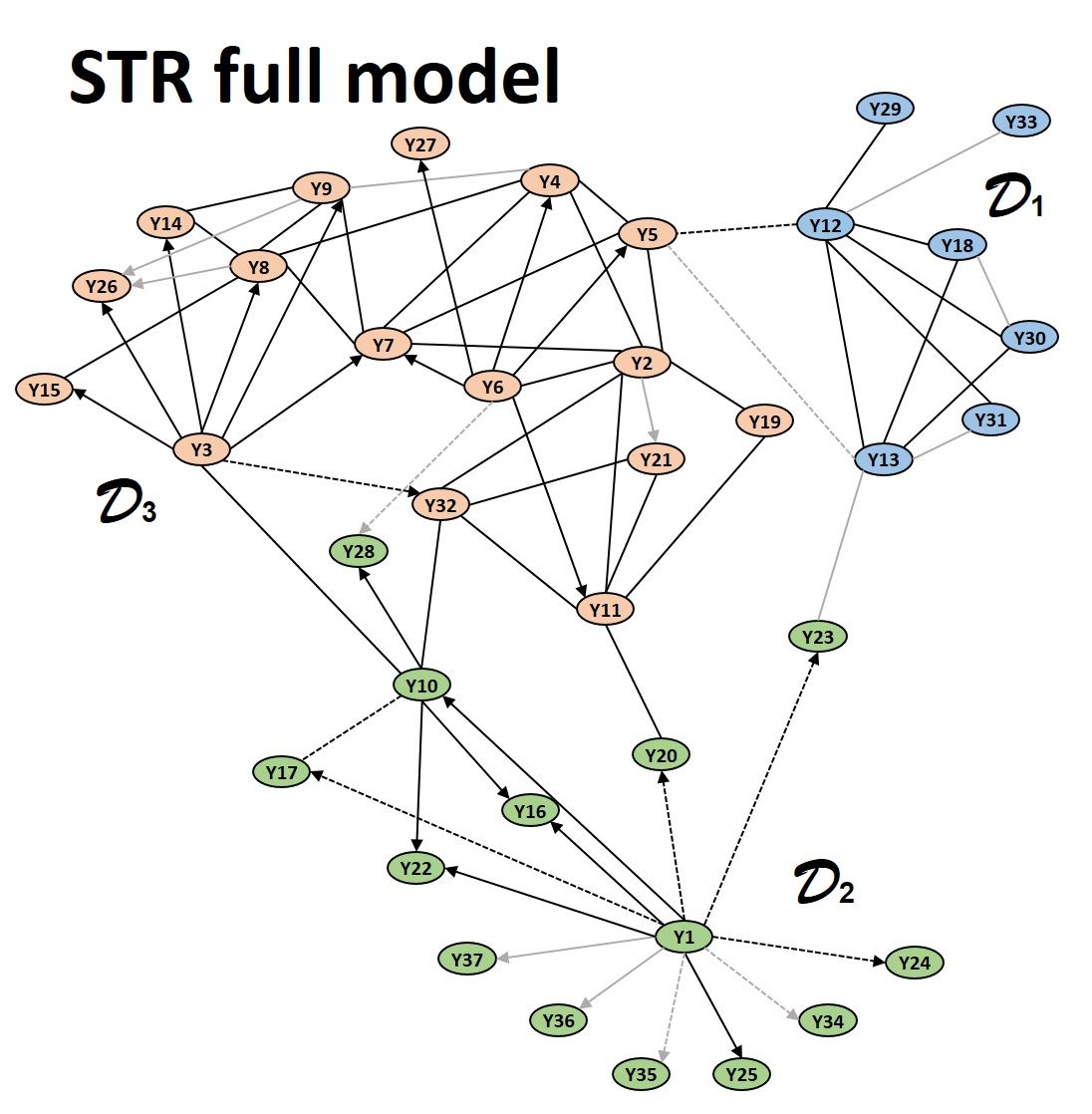}
  \end{minipage}
\end{center}
\caption{Digraphs of the NRP and STR full model. The vertices represent the genes and the arcs represent the interactions between them. Arcs with no direction indicate bi-directional interaction.}
\label{fig:NRPSTR}
\end{figure}

A vertex set partition, based on putative gene function and role in dormancy/latency development, defines three subdigraphs in both models: $\mathcal{D}_1,\mathcal{D}_2$ and $\mathcal{D}_3$. Subdigraph $\mathcal{D}_1$, is constituted by genes that are involved in the bacilli's cell wall process. Genes that have functions related to the DosR-regulon and adaptation are clustered in $\mathcal{D}_3$ and several genes that have other functions constitute $\mathcal{D}_2$. An arc is deleted between a vertex belonging to different subdigraphs.  These subdigraphs correspond to the subsystems of the $S$-system for the NRP and STR models.  

In this paper, we used CRN to represent the corresponding NRP and STR $S$-systems of \cite{MAMU2013} and study interesting network properties it exhibit. The CRN representation of each $S$-system is referred to as the \enquote{full network} $\mathscr{N}$ with the genes (variables) as species and each subsystem is referred to as the three embedded networks $\mathscr{N}_1^*,\mathscr{N}_2^*$ and $\mathscr{N}_3^*$. These embedded networks resulted from the vertex (species) set partition of \cite{MAMU2013}.

The CRNs for the NRP and STR models as well as the three embedded networks in each model are documented in the Supplementary Information. Table \ref{tab:NRP and STR network numbers N and N*i } summarizes the network numbers for the CRN representations of the S-system of the NRP and STR full models $\mathscr{N}$ and their corresponding embedded networks $\mathscr{N}_i^*$.

\begin{table}[h!]
\footnotesize
\begin{center}
\caption{NRP and STR Network numbers for $\mathscr{N}$ and its embedded networks $\mathscr{N}_i^*$}
\label{tab:NRP and STR network numbers N and N*i }       
\begin{tabular}{ccccc|cccc}
\hline\noalign{\smallskip}
\multirow{2}{*}{Network numbers} &
      \multicolumn{4}{c}{NRP} & 
      \multicolumn{4}{c}{STR} \\
    & $\mathscr{N}$ & $\mathscr{N}_1^*$ & $\mathscr{N}_2^*$ & $\mathscr{N}_3^*$ &	 $\mathscr{N}$ & $\mathscr{N}_1^*$ & $\mathscr{N}_2^*$ & $\mathscr{N}_3^*$ \\     
\noalign{\smallskip}\hline\noalign{\smallskip}
$m$ 			& 40 & 8  & 17 & 15	& 37 & 7  & 14 & 16 \\
$m_{rev}$		& 0  & 0  &  0 &  0 & 1  & 0  & 1  & 1 \\
$n$				& 98 & 19 & 36 & 42 & 95 & 19 & 29 & 44\\
$n_r$			& 60 & 10 & 20 & 28 & 60 & 12 & 17 & 29\\
$r$				& 80 & 16 & 34 & 30 & 74 & 14 & 28 & 32\\
$r_{irrev}$		& 80 & 16 & 34 & 30	& 72 & 14 & 26 & 30\\
$l$				& 19 & 3  &  3 & 12 & 22 & 5  &  2 & 13 \\
$sl$			& 98 & 19 & 36 & 42 & 94 & 19 & 28 & 43\\
$t$				& 38 & 9  & 16 & 14 & 35 & 7  & 12 & 15\\
$s$				& 40 & 8  & 17 & 15 & 37 & 7  & 14 & 16\\
$q$				& 40 & 8  & 17 & 15 & 37 & 7  & 14 & 16\\
$\delta$		& 39 & 8  & 16 & 15 & 36 & 7  & 13 & 15\\
$\delta_p$		& 20 & 2  & 3  & 13 & 23 & 5  & 3  & 13\\
\noalign{\smallskip}\hline
\end{tabular}
\end{center}
\end{table}

\section{The deficiency formula for S-system CRNs}
\label{sec:4}
In \cite{AJMM2015}, Arceo et al. introduced the concepts of reversibility and non-self regulation of dependent species (variables) of an $S$-system and derived a formula for calculating the total deficiency of a non-self regulating (NSR) S-systems (Proposition 10 of \cite{AJMM2015}). It turns out that the equality in the formula does not hold in general. In this Section, we present two counterexamples: the NRP full network $\mathscr{N}$ and its embedded network $\mathscr{N}_2^*$. We restate the result and provide a proof based on decomposition theory. 

First, we recall definitions from \cite{AJMM2015}.

\begin{definition}
An $S$-system is called \textbf{non-self regulating} (NSR) if it is monomolecular and each dependent species $X_i$ does not occur in its production term.
\end{definition}

\begin{remark}
The term \enquote{monomolecular} refers to the biochemical map or digraph underlying the S-system. 
\end{remark}

\begin{definition}
Let $R_j$ and $P_j$ be the set of variables regulating the inflow and outflow reactions of the species $X_j$ (= dependent variable) of an $S$-system, respectively. The species is called \textbf{reversible} if $R_j=P_j$. Otherwise, it is called irreversible. An $S$ -system is called reversible (irreversible) if all its species are reversible (irreversible).
\end{definition}

The term reversible species comes from the fact that the reaction pair associated with $X_j$ is $R_j\leftrightarrow Xj + Rj$. More generally, we write $R_j\rightarrow X_j + R_j$ and $X_j + P_j \rightarrow P_j$ for the
reaction pair associated with a species.

\begin{example}
Both the NRP and STR full models $\mathscr{N}$ and their embedded networks $\mathscr{N}^*_i$ are non-self regulating S-sytem. But only the NRP $\mathscr{N}$ and its embedded networks $\mathscr{N}^*_i$ are irreversible. The STR $\mathscr{N}$ and some of its embedded networks $\mathscr{N}^*_i$ have a reversible species.
\end{example}

Next, we state the formula in computing for the total deficiency of an NSR $S$-system. 

\begin{proposition}
\label{prop:prop 10 in AJMM2015}
(Proposition 10 in \cite{AJMM2015}). Let $m$ be the number of dependent variables (species) in an NSR $S$-system. If $m_{rev}$ is the number of reversible species then the total deficiency is $m-m_{rev}$.
\end{proposition}

As mentioned earlier, the equality in the formula does not hold in general. Computations for the two counterexamples of this proposition show that in  NRP $\mathscr{N}$ it has $m=40$ species (all of which are dependent variables), no reversible ones ($m_{rev}=0$) but deficiency $\delta=39 < 40-0$. Also, $\delta=13<14-0$ in NRP $\mathscr{N}^*_2$ (see Table \ref{tab:NRP and STR network numbers N and N*i }).


A decomposition theory approach is used to restate and provide a new proof of this Proposition. Fundamental concepts and known results of decomposition theory are discussed in Section \ref{decomposition}. But before presenting the reformulated deficiency formula for $S$-system CRNs, we first introduce the new concept of  \enquote{incidence-independence} and derive its basic properties. Incidence-independence is a natural complement to the\enquote{independence} property of a decomposition, but has so far been neglected in the CRNT literature.

\begin{definition}
A decomposition \{$\mathscr{N}_1,...,\mathscr{N}_k\}$ of a CRN is \textbf{incidence-independent} if and only if the image of the incidence map of $\mathscr{N}$ is the direct sum of the images of the incidence maps of the subnetworks. 
\end{definition} 

Since the direct sum property of the images is equivalent to the dimension of the image of $I_a$ $=$ the sum of the dimensions of the subnetwork incident map images, an equivalent formulation is the following equality: 
\begin{equation}
\label{eqn:incidence-independence}
n-l = \sum(n_i-l_i).
\end{equation}

\begin{example}
The linkage classes form the primary example of an incidence-independent decomposition, since $n=\sum n_i$ and $l=\sum l_i$.
\end{example}

We have the following proposition which is analogous to Proposition \ref{prop: independent coarsening}:

\begin{proposition}
\label{prop:incidence-independent coarsening}
If a decomposition is incidence-independent, then any coarsening of the decomposition is incidence-independent.
\end{proposition} 

The proof has the same argumentation as in Proposition \ref{prop: independent coarsening} now applied to the image of the incidence map instead of the stoichiometric subspace.

The following proposition is a familiar property of linkage classes:

\begin{proposition}
\label{prop:deficiency of incidence independent}
For an incidence independent decomposition $\mathscr{N}=\mathscr{N}_1 +...+\mathscr{N}_k$, then $\delta \geq \delta_1+...+\delta_k$.
\end{proposition}

\begin{proof}
Since for any decomposition, $s \leq \sum s_i$, subtracting the LHS from $n-l$ and the RHS from $\sum n_i -\sum l_i$ delivers the claim.
\end{proof}

\begin{definition}
A decomposition is \textbf{bi-independent} if it is both independent and incidence-independent.
\end{definition}

The ILC property of linkage classes is the best known example of a bi-independent decomposition.

\begin{corollary}
\label{cor:deficiency of bi-independent}
For a bi-independent decomposition, $\delta = \delta _1 +...+ \delta _k$.
\end{corollary}

\begin{proof}
It was shown in \cite{FLRM2017}, that for an independent decomposition: $\delta \leq \delta _1 +...+ \delta _k$. Hence claim follows directly from the previous proposition.
\end{proof}

Another new concept that emerged in our restatement of the deficiency formula for $S$-system embedded networks is the species decomposition.

\begin{definition}
Let $R_i$ and $P_i$ be the set of variables regulating the inflow and outflow reactions of the species $X_i$ (= dependent variable) of an $S$ -system, respectively. A set of subnetworks $\{\mathscr{N}_i=(\mathscr{S}_i, \mathscr{C}_i,\mathscr{R}_i), i=1,...,m\}$ is a \textbf{species decomposition} of an $S$-system $\mathscr{N}$ if the subnetwork is formed by the pair of reactions associated with each species, i.e,  $\mathscr{R}_i=\{R_i \rightarrow X_i+R_i, X_i +P_i \rightarrow P_i\}$.
\end{definition}

\begin{remark}
A species decomposition is another way of decomposing a network since it partitions the reaction set. Also, in this decomposition, a network with $m$ species results into $m$ subnetworks. 
\end{remark}

\begin{example}
\label{example:species decomposition}
Consider an S-system with 3 species $\{X_1, X_2, X_3\}$. Since each species corresponds to a pair of reactions, we have 6 reactions in total. Suppose we have the following CRN $\mathscr{M}$: $\{X_1 \rightarrow 2X_1, X_1 +X_2 +X_3 \rightarrow X_2 +X_3,
X_1+X_2 \rightarrow X_1 + 2X_2, X_1+X_2 \rightarrow X_1,
X_2 +X_3 \rightarrow 2X_3+X_2, X_3 \rightarrow 0\}$. We have the following species decomposition (subnetworks) of $\mathscr{M}$:

$\mathscr{M}_1=\{X_1 \rightarrow 2X_1, \quad X_1 +X_2 +X_3 \rightarrow X_2 +X_3\}$

$\mathscr{M}_2=\{X_1+X_2 \rightarrow X_1 + 2X_2,\quad  X_1+X_2 \rightarrow X_1\}$

$\mathscr{M}_3=\{X_2 +X_3 \rightarrow 2X_3+X_2,\quad  X_3 \rightarrow 0\}$
\end{example}

Finally, we now reformulate the deficiency formula for $S$-system CRNs. The correct formula is an inequality, instead of an equality. 

To avoid the minor complications of independent variables, we consider the embedded representation of an $S$-system. In particular, the species set $\mathscr{S'}$ consists only of the dependent variables $\{X_1. ..., X_m\}$.

\begin{theorem}
Let $\mathscr{N'}=(\mathscr{S'},\mathscr{C'},\mathscr{R'})$ be the embedded representation of an $S$-system and $\mathscr{R}'_i$ denotes the set $\{R_i \rightarrow X_i + R_i, X_i + P_i \rightarrow P_i\}$
\begin{enumerate}[i)]
\item The sets $\{\mathscr{R}'_i\}$ form a partition of the reaction set $\mathscr{R}'$ and hence induce a  species decomposition.
\item The species decomposition is independent and implies that $\delta \leq m - m_{rev}$. If the species decomposition is also incidence-independent, i.e. bi-independent, then $\delta=m-m_{rev}$. 
\end{enumerate}
\end{theorem}

\begin{proof} 
i) We have to show that if $i\neq j$, then the intersection of the reaction sets is empty.
Suppose that the sets $\{\mathscr{R}'_i\}$ do not form a partition of the reaction set $\mathscr{R}'$. Then there exists two sets $\mathscr{R}'_i$ and $\mathscr{R}'_j$, where $i \neq j$, that has a common reaction. We consider the following cases: a) two inflow reactions coincide and b) an inflow reaction coincides with an outflow reaction. The remaining cases involve converse reactions and hence follow similarly.

We let $\mathscr{R}'_i=\{R_i \rightarrow X_i+R_i, X_i + P_i \rightarrow P_i\}$ and $\mathscr{R}'_j=\{R_j \rightarrow X_j+R_j, X_j + P_j \rightarrow P_j\}$. We denote the subvectors of $R_i, R_j, P_i$ and $P_j$ as  $V_i, V_j, W_i$ and $W_j$, respectively. We set $V_i=(a_1,...,a_m)$ and $W_j=(b_1,...,b_m)$. In connection to the $a$ elements of $V_i$, the two input reactions in $\mathscr{R}'_i$ and $\mathscr{R}'_j$ coincide thus $R_i=R_j$ and $X_i +R_i=X_j +R_j$. This implies that $V_i=V_j$ and $X_i + V_i=X_j+V_i$ or $(a_1,...,a_i+1,...,a_m)=(a_1,...,a_j+1,...,a_m)$. Since $i \neq j$, $a_i+1=a_i$ and $a_j=a_j+1$, a contradiction.

As for the $b$ elements of $W_j$, we assume that an inflow reaction in $\mathscr{R}'_i$ coincides with an $\mathscr{R}'_j$. Then $R_i=X_j+P_j$ and $X_i +R_i=P_j$. Thus, we have $V_i=X_j+W_j$ or $(a_1,...,a_j,...,a_m)=(b_1,...,b_j+1,...,b_m)$. This implies that $a_i=b_i$ and $a_j=b_j+1$. Similarly, $V_i+X_i=W_j$ implies that $a_i+1=b_i$ and $a_j=b_j$. Since $i \neq j$, $a_i=b_i=a_i+1$ and $b_j=a_j=b_j+1$, a contradiction.

ii) Note that the stoichiometric subspace $S_i$ of each of the $m$ subnetworks $\mathscr{Ri}'_i$ of the species decomposition is $\{X_i\}$. Thus, the rank of each $\mathscr{Ri}'_i$ is 1. Since there are $m$ subnetworks and the rank of an $S$-system is m, $s=m=s_1+...+s_m$ and this implies independence.  In a species decomposition, the subnetworks are either reversible or irreversible. Each reversible species generates 2 distinct complexes and 1 linkage class, while an irreversible one generates 4 distinct complexes and 2 linkage classes. Hence, $\delta_i=0$ if $X_i$ is reversible and $\delta_j=1$ if $X_j$ is irreversible. Since the species decomposition is independent, from Proposition 3.3 of \citep{FLRM2017}, we have $\delta \leq \delta_1 +...+\delta_m = m- m_{rev}$. From Corollary \ref{cor:deficiency of bi-independent}, if the decomposition is incidence-independent too, then equality holds.
\end{proof}

\begin{example}
The following table shows  some of the network numbers for the species decomposition of $\mathscr{M}$ in Example \ref{example:species decomposition}. 

\begin{table}[h!]
\footnotesize
\begin{center}   
\begin{tabular}{ccccc}
\hline\noalign{\smallskip}
Network Numbers & $\mathscr{M}$ & $\mathscr{M}_1$ & $\mathscr{M}_2$ & $\mathscr{M}_3$  \\
\noalign{\smallskip}\hline\noalign{\smallskip}
$m$				& 3	& 3 &	2 & 2 \\
$n$				& 9 & 4 &	3 & 4 \\
$l$				& 3 & 2 &  1 & 2 \\
$s$				& 3 & 1  & 1 & 1 \\
$\delta$			& 3 & 1  & 1 & 1 \\
\noalign{\smallskip}\hline
\end{tabular}
\end{center}
\end{table}
\noindent
The species decomposition of $\mathscr{M}$ is bi-independent since $s=\sum s_i$ and $\delta = \sum \delta_i$. Also, $\mathscr{N}$ has no reversible species thus, $\delta=3=m-m_{rev}$.
\end{example}

The following corollary provides a sufficient condition for the species decomposition of an $S$-system to be incidence-dependent:

\begin{corollary}
If $\{\mathscr{S}_1,...,\mathscr{S}_k\}$ is a partition of the species set, then the sets $R(\mathscr{S}_j):=\{\mathscr{R}_i$ with $X_i$ in $\mathscr{S}_j\}$ form a partition of the reaction set and hence induce an independent decomposition of the network. The deficiency of the network is less than or equal to the sum of the deficiencies of the subnetworks. If the deficiency is less than the sum, then the species decomposition of the network is incidence-dependent.
\end{corollary}

\begin{example}
\label{Ex:species decomposition m-mrev}
Since for the full network $\mathscr{N}$ of the NRP model, deficiency $\delta=39<40=m-m_{rev}$, the species  decomposition of the network is incidence-dependent. The species decomposition of the NRP embedded representation $\mathscr{N}^*_2$ is also incidence-dependent since its deficiency  $\delta=14<13=m-m_{rev}$.
\end{example}


\section{Properties of the S-system CRNs and their embedded networks}
\label{sec:5}

We discuss some important network properties of NRP and STR full networks and their embedded subnetworks. In Section \ref{sec:5.1}, we consider their point terminality (PT), sufficient reactant diversity (SRD), terminality boundedness by deficiency (TBD) and equality of reactant and stoichiometric subspaces (RES) as well as the relationships between these characteristics. In Section \ref{sec:5.2} , we derive their discordance from a general result about S-system networks.

\subsection{Network properties of the NRP and STR models}
\label{sec:5.1}

The NRP and STR  full models and their embedded networks are fully open networks since each species has an outflow reaction as seen on their corresponding CRN in the Supplementary Information. Hence, the stoichiometric subspace $S=\mathbb{R}^\mathscr{S}$ thus, the rank of the network is equal to the number of species, i.e $s=m$ as shown in the first column of Tables \ref{tab:NRP network properties } and \ref{tab:STR network properties}. 

\begin{table}[h!]
\footnotesize
\caption{NRP Network properties.}
\label{tab:NRP network properties }  
\centering     
\begin{tabular}{cccccccccc}
\hline\noalign{\smallskip}
\textbf{Model} & $s = m$ & \textbf{WR} & \textbf{TM} &	\textbf{ET} & \textbf{ILC} & \textbf{PT} & \textbf{SRD} & \textbf{TBD} & \textbf{RES}\\
\noalign{\smallskip}\hline\noalign{\smallskip}
$\mathscr{N}$		& 40 & NO & NO & NO	&	NO & YES & YES & YES & YES \\
$\mathscr {N}_1^*$	& 8 & NO & NO & NO 	&	NO	& YES & YES & YES & YES \\
$\mathscr {N}_2^*$	& 17 & NO & NO & NO &	NO	& YES & YES & YES & YES \\
$\mathscr {N}_3^*$	& 15 & NO & NO & NO &	NO	& YES & YES & YES & YES \\
\noalign{\smallskip}\hline
\end{tabular}
\end{table}

\begin{table}[h!]
\footnotesize
\caption{STR Network properties}
\label{tab:STR network properties} 
\centering      
\begin{tabular}{cccccccccc}
\hline\noalign{\smallskip}
\textbf{Model} & $s = m$ & \textbf{WR} & \textbf{TM} & \textbf{ET} & \textbf{ILC} & \textbf{PT} & \textbf{SRD} & \textbf{TBD} & \textbf{RES}\\
\noalign{\smallskip}\hline\noalign{\smallskip}
$\mathscr{N}$		& 37 & NO & NO	&	X 	& NO & YES & YES & YES & YES \\
$\mathscr {N}_1^*$	& 7 & NO & NO 	&	NO	& NO & YES & YES & YES & YES \\
$\mathscr {N}_2^*$	& 14 & NO & NO & 	X	&	NO & YES & YES & YES & YES \\
$\mathscr {N}_3^*$	& 16 & NO & NO & 	X &	NO	& YES & YES & YES & YES \\
\noalign{\smallskip}\hline
\end{tabular}
\end{table}

In a weakly reversible (WR) network, $sl=l$. The network numbers of the NRP and STR models in Table \ref{tab:NRP and STR network numbers N and N*i } shows that $sl \neq l$ hence, the networks are not weakly reversible. Since weakly reversible CRNs are t-minimal ($t=l$), the \enquote{NO's} in the third column follow from those in the second column of Tables \ref{tab:NRP network properties } and \ref{tab:STR network properties}. 

We used CoNtRol, an online toolbox to test the endotacticity (ET) of the network. Due to the complexity of some networks, ContRol was not able to generate the said test. That is why the ET column in Tables \ref{tab:NRP network properties } and \ref{tab:STR network properties} is incomplete. We also used another software tool, ERNEST, consisting of extensible Matlab code which is specifically valuable for large networks such as the NRP and STR to determine the independent linkage class (ILC) property and the network numbers in Table \ref{tab:NRP and STR network numbers N and N*i }. From the report, it showed that the rank of the networks are not equal to the sum of the rank of their corresponding linkage classes. Thus, the \enquote{NO's} in the ILC column.

The WR, TM and ET columns in Tables \ref{tab:NRP network properties } and \ref{tab:STR network properties} are well-studied network properties in the current CRNT literature. Meanwhile, the point terminal (PT) and terminality bounded by deficiency (TBD) columns were introduced in \cite{AJLM2017} as additional network properties in their 15 BST case studies. They defined a network to be point terminal if $t=n-n_r$. For the TBD, we first introduce some convenient notation:

\begin{definition}  
The \textbf{terminality} $\tau(D)$ (or simply $\tau$) of a digraph $D$ is the non-negative integer $t - l$.
\end{definition}

In other words, a $t$-minimal digraph  is one with zero terminality and a non-$t$-minimal one with positive terminality.

\begin{definition}
A CRN is of type \textbf{terminality bounded by deficiency} (\textbf{TBD}) if $t - l \leq \delta$, otherwise, of type \textbf{terminality not deficiency-bounded} (\textbf{TND}), i.e. $t - l > \delta$. 
\end{definition}

A remark of Feliu and Wiuf \cite{FW2012} that $q < s$ (i.e. low reactant rank) implied degeneracy of all positive equilibria of a mass action system led us to the concept of \enquote{sufficient reactant diversity} (SRD) as a necessary condition for the existence of a non-degenerate equilibrium $(q \geq s \Rightarrow n_r \geq s)$.

\begin{definition}
A \textbf{CRN} has \textbf{low reactant diversity} (\textbf{LRD}) if $n_r < s$, otherwise it has \textbf{sufficient reactant diversity} (\textbf{SRD}). An SRD network has \textbf{high reactant diversity} (\textbf{HRD}) or \textbf{medium reactant diversity} (\textbf{MRD}) if $n_r > s$ or $= s$, respectively.
\end{definition}

The next Proposition clarifies the relationship between deficiency-bounded terminality and sufficient reactant diversity.

\begin{proposition}\label{TBDSRDprop}
Let $\mathscr N$ be a chemical reaction network.
\begin{enumerate}
\item [(i)] A network with deficiency-bounded terminality (TBD) has sufficient reactant diversity.
\item [(ii)] If the network is point terminal, then the converse also holds, i.e. TBD $\Leftarrow$ SRD (or equivalently TND $\Rightarrow$ LRD).
\end{enumerate}
\end{proposition}
\begin{proof}
From $\delta - \tau(\mathscr N) = n - l - s - (t - l) = n - t - s$, we obtain TBD, i.e. $\delta – \tau(\mathscr N ) \geq 0$ iff  $n - t \geq s$. Since $n_r \geq n - t$ for any network, TBD implies SRD. If $\mathscr N$ is point terminal $n_r = n- t$, showing that the converse also holds.  
\end{proof}

\begin{corollary}
A $t$-minimal network is an SRD network.
\end{corollary}

\begin{remark}
We have the following ascending chain of networks:  reversible $\Rightarrow$ weakly reversible $\Rightarrow$ $t$-minimal $\Rightarrow$ TBD $\Rightarrow$ SRD. 
\end{remark}

\begin{example}
Since all the $S$-system embedded networks in Tables \ref{tab:NRP network properties } and \ref{tab:STR network properties} are point terminal, Theorem~\ref{thm3}.(i) tells us that the information in the TBD and SRD columns are the same.
\end{example}

In \cite{AJLM2018}, a classification of CRNs based on the intersection $R \cap S$ of the reactant and stoichiometric subspaces was introduced. Open networks belong to the network class SRS (stoichiometry-determined reactant subspace) with $R \cap S = R$, since $R \subset S = \mathbb R^\mathscr S$. Tables \ref{tab:NRP network properties } and \ref{tab:STR network properties} show that all $S$-system embedded networks belong to the subset of coincident $R$ and $S$ subspaces, which we denote by RES ($R$ Equals $S$). 

We observe that, for open networks, $R = S$ is equivalent to $q = s$, i.e. the network's rank difference $\Delta(\mathscr N ) = 0$. 

We include the statements of the Proposition \ref{TBDSRDprop} to collect the relationships between TBD, SRD and RES in the following Theorem:

\begin{theorem}\label{thm3}
\begin{enumerate}
\item [(i)] A TBD network has sufficient reactant diversity. If the network is point terminal, then the converse holds, i.e. TBD $\Leftrightarrow$  SRD.
\item [(ii)] A TPD network has high reactant diversity (HRD). If the network is point terminal, then the converse also hold, i.e. TPD $\Rightarrow$ HRD.
\item [(iii)] An RES network has sufficient reactant diversity.
\item [(iv)] RES $\cap$ TBD = RES $\Leftrightarrow n_r \geq s + t_c$ . If the network is point terminal, then RES implies TBD. If the network is not point terminal, then it has high reactant diversity $(n_r > s)$ and reactant deficiency $\delta_\rho > 0$. 
\end{enumerate}
\end{theorem}

\begin{proof}
\begin{enumerate}
\item [(i)] was already shown in Proposition \ref{TBDSRDprop}.
\item [(ii)] $t-l<\delta\Rightarrow s<n-t\Rightarrow s<n_r-t_c \Rightarrow n_r>s$, since $t_c\geq 0$. If the network is point terminal, $t_c=0$, and the converse holds.
\item [(iii)] For any CRN, we have $n_r\geq q$. Since the CRN is RES, $q = s$, which establishes the sufficient reactant diversity. 
\item [(iii)] $\delta - \tau(\mathscr N ) = n - t - s = n - (t_c + t_p) - s = n - (t_c + n - n_r) - s = n_r - q - t_c = n_r - (s + t_c)$.  Hence TBD $\Leftrightarrow$  $n_r \geq s + t_c$, so that the first claim follows. For point terminal networks $t_c = 0$, and (ii) establishes the validity of the RHS. If the network is not point terminal, $t_c > 0$, hence $n_r > s$ and $\delta_\rho : = n_r - q = n_r - s > 0$. 
\end{enumerate}
\end{proof}

\begin{remark}
The inequality $n_r \geq s + t_c$ expresses clearly which additional characteristic an RES network needs--beyond SRD--to imply deficiency-bounded terminality.
\end{remark}

Figure \ref{ssystemembedded} illustrates the relationships for $S$-system embedded networks.

\begin{figure}[h!]
\begin{center}
\includegraphics[width=0.3\textwidth]{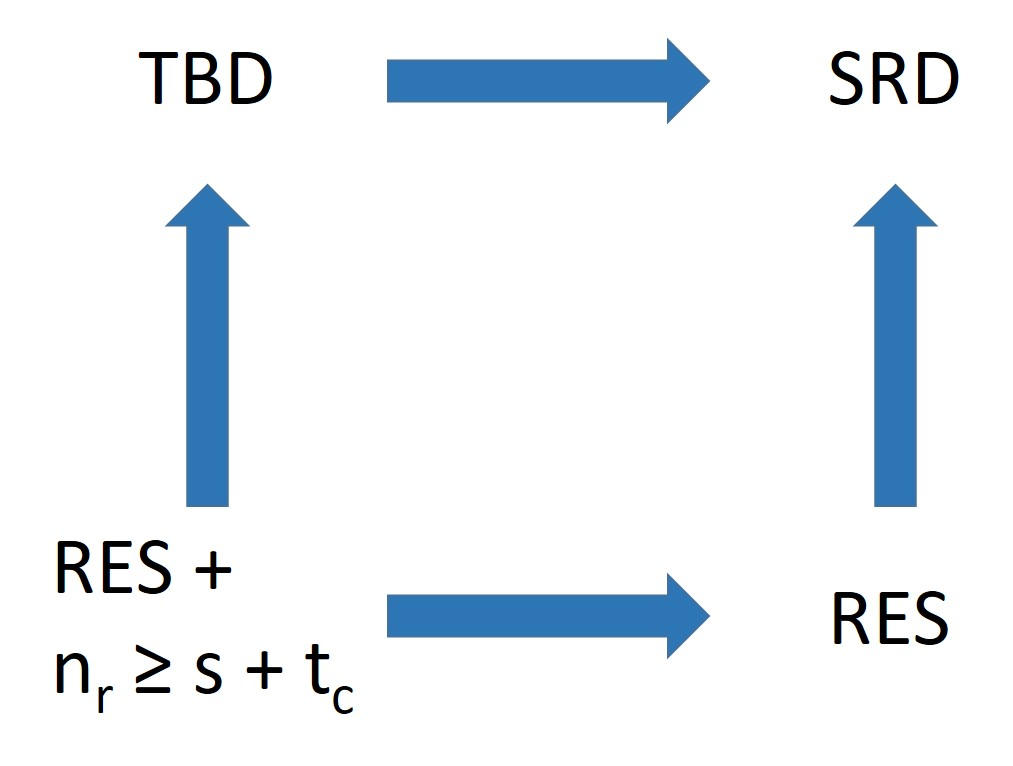} 
\caption{Relationships for S-system embedded networks}
\label{ssystemembedded}
\end{center}
\end{figure}

These network properties of the NRP and STR show a very high degree of coincidence with the eight embedded S-systems from BST Case Studies in \cite{AJLM2017} as seen in Table \ref{table: 8 embedded}. 

\begin{table}[h!]
\footnotesize
\caption{Network properties of $S$-system embedded networks of the BST Case Studies.}
\label{table: 8 embedded}
\centering
\begin{tabular}{cccccccccc}
\hline\noalign{\smallskip}
\textbf{Model}&$s=m$& \textbf{WR}&\textbf{TM}&\textbf{ET} &\textbf{ILC} &\textbf{PT}& \textbf{SRD} &\textbf{TBD} &\textbf{RES} \\
\noalign{\smallskip}\hline\noalign{\smallskip}
ARL1-S & 17 & NO & NO & NO  & NO   & YES & YES & YES & YES \\
ARL3-S & 4 	& NO & NO & NO & NO    & YES & YES & YES & YES \\
ARL4-S & 7 	& NO & NO & NO & NO   & YES & YES & YES & YES \\
CPA3-S & 9 	& NO & NO & NO & NO   & YES & YES & YES & YES \\
CPA4-S & 5 	& NO & NO & NO & NO  & YES & YES & YES & YES \\
ECJ0-S & 16 &NO & NO & NO & NO  & YES & YES & YES & YES \\
ECJ2-S & 5 & NO & NO & NO & NO  & YES & YES & YES & YES \\
ERM0-S & 5 & NO & NO & NO & NO  & YES & YES & YES & YES \\
\noalign{\smallskip}\hline
\end{tabular}
\end{table}

\subsection{The discordance of S-system CRNs}
\label{sec:5.2}

Limitations in the software tools available for Concordance tests in Tables \ref{tab:NRP network properties } and \ref{tab:STR network properties}  prevented results regarding this property. However, in this section, we present the discordance of $S$-system CRNs. Throughout this section, we consider both total and embedded $S$-system reaction networks and their kinetics.

\subsubsection{Concordance, weakly monotonic kinetics and non-inhibitory Power Law Kinetics}
\label{sec:5.2.1}

We recall from \cite{AJMM2015} the definition of concordance and weakly monotonic kinetics. We then complete the proof of the basic result relating both concepts to injectivity in non-inhibitory power law kinetics.

In preparation for the definition of concordance, we consider a reaction network $(\mathscr{S}, \mathscr{C}, \mathscr{R})$ with stoichiometric subspace $S \subset \mathbb{R}^\mathscr{S}$ and we let $L : \mathbb{R}^\mathscr{R} \rightarrow S$ be the linear map defined by

\begin{equation}
\label{eqn:linear map}
L\alpha = \sum _{y \rightarrow y' \in \mathscr{R}} \alpha _{y \rightarrow y'} (y' - y).
\end{equation}

Note that the real scalar multipliers $\{\alpha_{y \rightarrow y'}\}_{y \rightarrow y' \in \mathscr{R}}$ in Equation (\ref{eqn:linear map}) are permitted to be positive, negative or zero.

\begin{definition}
The reaction network $(\mathscr{S}, \mathscr{C}, \mathscr{R})$ is \textbf{concordant} if there do not exist an $\alpha \in ker L$ and a nonzero $\sigma \in S$ having the following properties:
\begin{enumerate} [i)]
\item For each $y \rightarrow y' \in \mathscr{R}$ such that $\alpha_{y \rightarrow y'} \neq 0$, supp $y$ contains a species $s$ for which sgn $\sigma_s =$ sgn $\alpha_{y \rightarrow y'}$.
\item For each $y \rightarrow y' \in \mathscr{R}$ such that $\alpha_{y \rightarrow y'}=0$, either $\sigma_s = 0$ for all $s \in$ supp $y$ or else supp $y$ contains a species $s$ and $s'$ for which sgn $\sigma_s = -$ sgn $\sigma_{s'}$, both not zero.
\end{enumerate}
A network that is not concordant is \textbf{discordant}.
\end{definition}

The set of weakly monotonic kinetics is defined as follows:
\begin{definition}
\label{def:weakly monotonic}
A kinetics $\mathscr{K}$ for a reaction network $(\mathscr{S}, \mathscr{C}, \mathscr{R})$ is \textbf{weakly monotonic} if, for each pair of compositions $c^*$ and $c^{**}$, the following implications hold for each reaction $y \rightarrow y' \in \mathscr{R}$ such that supp $y \subset$ supp $c^*$ and supp $y \subset$ supp $c^{**}$: 
\begin{enumerate} [i)]
\item $\mathscr{K}_{y \rightarrow y'} (c^{**}) > \mathscr{K}_{y \rightarrow y'} (c^*) \Rightarrow$ there is a species $s \in$ supp $y$ with $c_s^{**} > c_s^*$.
\item $\mathscr{K}_{y \rightarrow y'} (c^{**}) = \mathscr{K}_{y \rightarrow y'} (c^*) \Rightarrow$  $c_s^{**} > c_s^*$ for all $s \in$ supp $y$ or else there are species $s$, $s'\in$ supp $y$ with $c_s^{**} > c_s^*$ and $c_{s'}^{**} < c_{s'}^*$.
\end{enumerate}
\end{definition}

The final concept we need, to formulate the basic result, is that of non-inhibitory Power Law kinetics:

\begin{definition}
A power law kinetics is in $\mathcal{PL-NIK}(\mathscr{N})$ (or has \textbf{non-inhibitory kinetics}) if
\begin{enumerate}[i)]
\item the kinetic order matrix $F$ is non-negative and
\item a kinetic order $f_{r,s} > 0$ iff the species $s$ is an element of supp $\rho(r)$.
\end{enumerate}
\end{definition}

The following theorem of Shinar and Feinberg \cite{SF2012} relates concordance and injectivity in weakly monotonic and non-inhibitory power law kinetics:

\begin{theorem}
\label{thm:concordant}
For a chemical reaction network $\{\mathscr{S,C,R}\}$, the following conditions are equivalent:
\begin{enumerate}[i)]
\item The network is concordant.
\item The network is injective against all weakly monotonic kinetics.
\item The network is injective against all non-inhibitory power law kinetics.
\end{enumerate} 
\end{theorem}

\begin{proof}
(i)$\rightarrow$(ii): This is shown in Proposition 4.8 of \cite{SF2012}. \\
(iii) $\rightarrow$ (i): Inspection of the proof of Proposition 4.9 (Appendix A of \cite{SF2012}) reveals that the kinetics constructed is actually a non-inhibitory power law kinetics, thus proving this implication.\\

We now confirm that a non-inhibitory power law kinetics is weakly monotonic, from which (ii) $\rightarrow$ (iii) follows. To show that a kinetic $\mathscr{K}$ in $\mathcal{PL-NIK} (\mathscr{N}$ fulfills (i) of Definition \ref{def:weakly monotonic}, suppose that for all species in supp $\rho(r), c^*_s \leq c^{**}_s$. Since $f_{r,s} > 0$, then $(c^*_s)^{f_{r,s}} \leq (c^{**}_s)^{f_{r,s}}$. Taking the product on both sides over all $s$ in supp $\rho (r)$, we obtain a contradiction. To show (ii), again suppose there is an $s$ with $c^*_s < c^{**}_s$. Since $f_{r,s} > 0$, then $(c^*_s)^{f_{r,s}} < (c^{**}_s)^{f_{r,s}}$ - this is where property (ii) of the $\mathcal{PL-NIK}$ definition comes in. Unless there is an $s'$ with $c^*_{s'} > c^{**}_{s'}$ and consequently $(c^*_{s'})^{f_{r,s'}} > (c^{**}_{s'})^{f_{r,s'}}$, one cannot obtain the product equality.
\end{proof}

\subsubsection{Discordance of $S$-system reaction networks}
\label{sec:5.2.2}
\begin{theorem}
\label{thm:discordant}
\begin{enumerate}[i)]
\item An $S$-system reaction network with 2 or more dependent species is discordant.
\item A non-constant $S$-system reaction network with one dependent species is discordant if $A \neq 0$ and the product of its rate constants ($\alpha$ and $\beta$) and kinetic orders ($g$ and $h$)  $\alpha\beta gh \neq 0$, otherwise concordant.
\end{enumerate}
\end{theorem}

\begin{proof}
(i) To show discordance, we only need to construct a $\mathcal{PL-NIK}$ kinetics with a non-injective SFRF on the reaction network. We consider the total representation of the $S$-system first. The SFRF of any PL kinetics maybe decomposed into an $f_D$ (containing all coordinate functions for dependent species) and $f_i$ with the coordinate functions for the independent species. Since each $f_{i,j}$ has the form $f_{i,j} = a_j > 0$, iit is clear that the SFRF is injective iff $f_D$ is injective. We will construct a $\mathcal{PL-NIK}$ kinetics whose $f_D$ is multistationary, implying that it is not injective. To achieve this, we choose the kinetic order matrix as follows:

\begin{itemize}
\item The kinetics is $\mathcal{PL-NIK}$.
\item The first column of the derived $A$ matrix consists only of zeros, i.e. $g_{i1}=h_{i1}$, for all $i=1,...,m$.
\item $A \neq 0$ (this is possible since the system has 2 or more dependent species).
\end{itemize}
We have an $S$-system with det $A = 0$, which according to the above, either has no steady state or infinitely many steady states. Since the rate constants can be arbitrarily chosen from the positive real numbers, we adjust them so that we have the latter alternative. Since the network is open, this shows that $f_D$ is not injective in the single stoichiometric class.

(ii)We begin this time with $A = 0$. Since we assumed it is not constant, then $\alpha \neq \beta$. So, the kinetics $= (\alpha - \beta) X^g = ( \alpha - \beta$ exp $(g ln X)$ is strictly monotonic,and hence injective for any positive $g$. Hence the CRNs are concordant according to Theorem ~\ref{thm:concordant}. If $A$ is nonzero, which is equivalent to det $A$ nonzero, the system has a unique positive equilibrium. If one of the rate constants or the kinetic orders is 0, then we obtain a function as in the $A = 0$ case plus a constant. Hence, the system is also injective. To show that it is not injective for both rate constants and both kinetic orders nonzero, we need to look at its derivative. Since the S -system in one variable has the form $f(x)=\alpha X^g - \beta X^h$ and $f'(x)=\alpha g X^{g-1} - \beta h X^{h-1}$. One easily checks that $f'(x) < (=)(>)0$ if and only if $x^{g-h} < (=)(>) \frac{\beta h}{\alpha g}$. From Calculus, it follows that the function is strictly decreasing and strictly increasing left and right of the particular
$x$-value. The value is between its steady state at 0 and its positive steady state at $\frac{\beta}{\alpha}$, since $g > h$ (i.e. production > degradation). This means that f is not injective in the interval $(0,\frac{\beta}{\alpha})$, so that it is discordant according to the Shinar-Feinberg Theorem in the previous section. A similar argument holds for $h > g$, the function is not injective in the interval $(0, \frac{\beta}{\alpha})$. For the constant case, the same argument as in (i) applies. 

These considerations also cover the embedded representation case, since the same $f_D$ describes the embedded system, by considering the powers of the independent species as part of the rate constants.
\end{proof}


\section{The $^*-$disjoint union $\mathscr{N}^*$ of the embedded networks of $\mathscr{N}$}
\label{sec:6}
In our study of the embedded networks of the NRP and STR full network, we observed that the zero complex was the only complex common to them. From this, we abstracted the concepts of a $\mathscr{C}^*$-decomposition and $\mathscr{C}$-decomposition. In this Section, structure theorems for both decomposition types are derived in terms of linkage classes. Moreover, a new $S$-system CRN $\mathscr{N}^*$ is also constructed for which the embedded networks $\mathscr{N}_i^*$ form a $\mathscr{C}^*$-decomposition.

\subsection{The  $\mathscr{C}^*$-decomposition and the $\mathscr{C}$-decomposition}
\label{sec:6.1}
 We first formulate a formal definition of a $\mathscr{C}^*$-decomposition:

\begin{definition}
A decomposition $\mathscr{N}=\mathscr{N}_1+\mathscr{N}_2+...+\mathscr{N}_k$ with $\mathscr{N}_i=(\mathscr{S}_i,\mathscr{C}_i,\mathscr{R}_i)$ is a \textbf{$\mathscr{C}^*$-decomposition} if $\mathscr{C}_i^* \cap \mathscr{C}_j^* =\emptyset $ for $i \neq j$ where $\mathscr{C}_i^*$ and $\mathscr{C}_j^*$ are the non-zero complexes in $\mathscr{C}_i$ and $\mathscr{C}_j$, respectively. 
\end{definition}

\noindent
The embedded networks $\mathscr{N}_i^*$ of the NRP and STR models are $\mathscr{C}^*$-decompositions of their full network $\mathscr{N}$, respectively.

A particularly interesting subset of $\mathscr{C}^*$-decompositions consists of the  $\mathscr{C}$-decomposition.

\begin{definition}
A decomposition $\mathscr{N}=\mathscr{N}_1+\mathscr{N}_2+...+\mathscr{N}_k$ with $\mathscr{N}_i=(\mathscr{S}_i,\mathscr{C}_i,\mathscr{R}_i)$ is a \textbf{$\mathscr{C}$-decomposition} if $\mathscr{C}_i \cap \mathscr{C}_j =\emptyset $ for $i \neq j$.
\end{definition}

A $\mathscr{C}$-decomposition partitions not only the set of reactions but also the set of complexes. The primary example of a $\mathscr{C}$-decomposition are the linkage classes. Linkage classes, in fact, essentially determine the structure of a $\mathscr{C}$-decomposition. We first present this Structure Theorem for a $\mathscr{C}$-decomposition: 

\begin{theorem}
\textbf{(Structure Theorem for $\mathscr{C}$-decomposition)}. Let $\mathscr{L}_1,...,\mathscr{L}_l$ be the linkage classes of a network $\mathscr{N}$. A decomposition $\mathscr{N}=\mathscr{N}_1+\mathscr{N}_2+...+\mathscr{N}_k$ is a $\mathscr{C}$-decomposition if and only if each $\mathscr{N}_i$ is the union of linkage classes and each linkage class is contained in only one $\mathscr{N}_i$. In other words, the linkage class decomposition is a refinement of $\mathscr{N}$.
\end{theorem}

\begin{proof}
Clearly, if the linkage classes form a refinement of $\mathscr{N}$, then $\mathscr{N}$ is a $\mathscr{C}$-decomposition. To see the converse, let $\mathscr{N}_i=(\mathscr{S}_i,\mathscr{C}_i,\mathscr{R}_i)$ and $\mathscr{L}_j=(\mathscr{S}_{\mathscr{L}_j},\mathscr{C}_{\mathscr{L}_j},\mathscr{R}_{\mathscr{L}_j})$ where $\mathscr{R}_i$ is the union (taken over $j$) of $(\mathscr{R}_i \cap \mathscr{R}_{\mathscr{L}_j}$. We only need to show that each non-empty intersection is equal to $\mathscr{R}_{\mathscr{L}_j}$, ( i.e., $\mathscr{R}_{\mathscr{L}_j}=\mathscr{R}_i \cap \mathscr{R}_{\mathscr{L}_j}$) to imply that each linkage class is contained in only one $\mathscr{N}_i$. If the linkage class $\mathscr{L}_j$ has only one reaction then $\mathscr{R}_i \cap \mathscr{R}_{\mathscr{L}_j}=\mathscr{R}_{\mathscr{L}_j}$. If the linkage class $\mathscr{L}_j$ has at least two reactions, then there is an adjacent reaction to each reaction, whose reactant complex or product complex is common with the first reaction. If this adjacent reaction belongs to a different subnetwork, then there exists a complex which is common to two different subnetworks. This would contradict that $\mathscr{N}$ partitions the set of complexes. Hence, all reactions of the linkage class lie in the intersection with $\mathscr{R}_i$.
\end{proof}

\begin{corollary}
For a $\mathscr{C}$-decomposition $\mathscr{N}=\mathscr{N}_1+\mathscr{N}_2+...+\mathscr{N}_k$, $k \leq l$. 
\end{corollary}

\begin{proof}
If $\mathscr{N}$ is decomposed according to linkage classes, then $\mathscr{N}_i=\mathscr{L}_i$. Thus, $k=l$. If each $\mathscr{N}_i$ is the union of linkage classes, then the number of subnetworks is less than the number of linkage classes. Hence, $k<l$.
\end{proof}

We also obtain a new characterization of the ILC property:

\begin{corollary}
A network has independent linkage classes if and only if every $\mathscr{C}$-decomposition is independent.
\end{corollary}

\begin{proof}
If a network has independent linkage classes then the stoichiometric subspace $S$ is the direct sum of the stoichiometric subspaces of the linkage classes. Grouping the summands  according to the unions of the linkage classes for the subnetworks of a $\mathscr{C}$-decomposition provides $S$ as the direct sum of the subnetworks. Hence, every $\mathscr{C}$-decomposition of a network is also independent. For the converse, since every $\mathscr{C}$-decomposition of a network is independent and the linkage class decomposition is also a $\mathscr{C}$-decomposition, it follows that the linkage classes are independent.
\end{proof}

\noindent
If a network has dependent linkage class, it may fail to have an independent $\mathscr{C}$-decomposition, as the following example shows:
\begin{example}
Consider the CRN with reactions $X_1 \rightarrow 2X_1 +X_2$ and $X_2 \rightarrow 2X_2 + X_1$, it has $\delta=1$. The only non-trivial decomposition is the linkage class decomposition, where the deficiency of the two linkage classes is 0. Clearly, the linkage class decomposition is dependent. In particular, it has no independent $\mathscr{C}$-decomposition. 
\end{example}

We now derive the Structure Theorem for $\mathscr{C}^*$-decomposition and apply this to the incidence-independence of such decomposition.

\begin{theorem}
\label{thm:STC*}
\textbf{(Structure Theorem for  $\mathscr{C}^*$-decomposition)}. Let $\mathscr{N}_1 + \mathscr{N}_2 + ... + \mathscr{N}_k$ be a $\mathscr{C}^*$-decomposition and $\mathscr{L}_0$ and $\mathscr{L}_{0,i}$ be the linkage classes of $\mathscr{N}$ and $\mathscr{N}_i$ containing the zero complex (note $\mathscr{L}_{0,i}$, is empty if $\mathscr{N}_i$ does not contain the zero complex). Then
\begin{enumerate}[i)]
\item the $\mathscr{L}_{0,i}$ form a $\mathscr{C}^*$-decomposition of $\mathscr{L}_0$
\item the (non-empty) $\mathscr{N}_i \setminus \mathscr{L}_{0,i}$ form a $\mathscr{C}$-decomposition of $\mathscr{N} \setminus \mathscr{L}_0$ 
\end{enumerate} 
\end{theorem}

\begin{proof}
To prove (i), we need to show that each non-zero complex of $\mathscr{L}_0$ is contained in only one subnetwork $\mathscr{N}_i$. If there is only one subnetwork $\mathscr{N}_i$ containing the zero complex then we are done. If there are at least two subnetworks containing the zero complex then $\mathscr{L}_0$ has at least two non-zero complexes connected to the zero complex. Otherwise, if there would only be one complex then $\mathscr{N}_i$ is not a $\mathscr{C}^*$-decomposition of $\mathscr{N}$, a contradiction. Now, if one of these non-zero complexes belongs to different subnetworks, this would contradict that $\mathscr{N}$ partitions the non-zero complexes. Hence, all the non-zero complexes of $\mathscr{L}_0$ is contained in only one $\mathscr{N}_i$ and $\mathscr{L}_{0,1} + \mathscr{L}_{0,2} +...+ \mathscr{L}_{0,j} = \mathscr{L}_0$ for $j \leq k$.

To prove (ii), it suffices to show that the intersection of the set of complexes in $\mathscr{N}_i \setminus \mathscr{L}_0$ is empty. The set of complexes in $\mathscr{N} \setminus \mathscr{L}_0$ are all non-zero and $\mathscr{N} \setminus \mathscr{L}_0 = (\mathscr{N}_1 + ...+ \mathscr{N}_k) \setminus \mathscr{L}_0$. From (i), we have $\mathscr{L}_{0,1} + \mathscr{L}_{0,2} +...+ \mathscr{L}_{0,j} = \mathscr{L}_0$ for $j \leq k$. Thus, $\mathscr{N} \setminus \mathscr{L}_0 = \mathscr{N}_1 \setminus \mathscr{L}_{0,1} + ... + \mathscr{N}_k \setminus \mathscr{L}_{0,k}$ where $\mathscr{L}_{0,k}$ is empty if $\mathscr{N}_k$ does not contain the zero complex. Since $\mathscr{N}_i$ is a $\mathscr{C}^*$-decomposition of $\mathscr{N}$, the intersection of the set of complexes in $\mathscr{N}_i \setminus \mathscr{L}_{0,i}$ is empty.
\end{proof}

Incidence-independence is given by Equation (\ref{eqn:incidence-independence}) in Section \ref{sec:4}. For $\mathscr{C}^*$-decompositions, this equation can be transformed into a more convenient \enquote{Common Complex Criterion} as follows: 

\begin{proposition}
\label{prop:CCC}
\textbf{(Common Complex Criterion).} Let $k(0)$ be the number of subnetworks $\mathscr{N}_i$ in a network $\mathscr{N}$ containing the zero complex. A $\mathscr{C}^*$-decomposition of $\mathscr{N}$ is incidence-independent if and only if
\[\sum l_i - l = \begin{cases} 
		\mbox{0,} & \mbox{whenever } k(0)=0 \\ 
		\mbox{k(0) - 1,} & \mbox{whenever } k(0)>0 
	\end{cases} \]
\end{proposition}

\begin{proof}
We recall Equation (\ref{eqn:incidence-independence}): 
$$n-l = \sum (n_i -l_i)$$ where $n_i$ and $l_i$ are number of complexes and linkage classes in the subnetwork $\mathscr{N}_i$. Obviously, if the network $\mathscr{N}$ does not contain the zero complex then $k(0)=0$. Since $\mathscr{N}$ is a $\mathscr{C}^*$-decomposition, $n = \sum n_i$. Thus,  Equation (\ref{eqn:incidence-independence}) is reduced to $\sum l_i -l =0.$ 

On the other hand, if $k(0)$ subnetworks contain the zero complex then, $n=n^*+1$ where $n^*$ is the number of non-zero complexes in $\mathscr{N}$. Also, $n_i = n_i^* +1$ where $n_i^*$ is the number of non-zero complexes in $\mathscr{N}_i$. Thus, Equation (\ref{eqn:incidence-independence}) becomes
$$n^*+1-l=\sum n_i^* +k(0) - \sum l_i.$$
Since $n^*=\sum n^*_i$, we obtain $\sum l_i - l = k(0)-1.$
\end{proof}

Using the parameter $k(0)$, we obtain two sufficient conditions for incidence-independence of $\mathscr{C}^*$-decomposition.

\begin{proposition}
\label{prop:k(0) for incidence-independence}
Let $k(0)$ be the number of subnetworks of a $\mathscr{C}^*$-decomposition containing the zero complex.
\begin{enumerate} [i)]
\item If $k(0)=0$ or $1$, then the decomposition is a $\mathscr{C}$-decomposition, and hence incidence-independent.
\item If $k(0)=k$, i.e. all subnetworks contain the zero complex, then the $\mathscr{C}^*$-decomposition is incidence-independent.
\end{enumerate}
\end{proposition}

\begin{proof}
$k(0)=0$ means that the network does not contain the zero complex, by definition, the decomposition is a $\mathscr{C}$-decomposition. This is also the case when the zero complex is in only one subnetwork in a $\mathscr{C}^*$-decomposition. The incidence-independence also follows from the Common Complex Criterion (CCC).

From the Structure Theorem of $\mathscr{C}^*$-decomposition (Theorem \ref{thm:STC*}), there are precisely $k$ subnetwork linkage classes in $\mathscr{L}_0$. Since the subnetworks $\mathscr{N}_i \setminus \mathscr{L}_{0,i}$ form a $\mathscr{C}$-decomposition of $\mathscr{N} \setminus \mathscr{L}_0$, and there is no zero complex in $\mathscr{N} \setminus \mathscr{L}_0$, $\sum (l_i -1) - (l-1)=0$ or $\sum l_i - l = k-1$, which is the CCC condition for incidence-independence of the $\mathscr{C}^*$-decomposition of the whole network.
\end{proof}

\subsection{The new $S$-system CRN $\mathscr{N}^*$}

As mentioned in Section \ref{sec:3}, the species set partition of the NRP and STR full models $\mathscr{N}$ generates embedded networks $\mathscr{N}^*_i$. But these $\mathscr{N}^*_i$'s does not partition the reaction set in $\mathscr{N}$. This led us in constructing a new $S$-system $\mathscr{N}^*$. 

\begin{definition}
Let $\mathscr{N}_i^*=(\mathscr{S}_i, \mathscr{C}|_{\mathscr{S}_i}, \mathscr{R}|_{\mathscr{S}_i})$, $i=1,...,k$, be the embedded networks induced by the species and reaction set partitions. The $^*$-disjoint union $\mathscr{N}^*$ of the embedded networks is the $S$-system CRN given by $(\mathscr{S}_i \cup \mathscr{C}|_{\mathscr{S}_i} \cup \mathscr{R}|_{\mathscr{S}_i})$.
\end{definition}

The adjective \enquote{$^*$-disjoint} stems from the fact that if the zero complex is present in the network, the union of all complexes may not be disjoint. However if $\mathscr{C}^*$ denotes the set of non-zero complexes, then their union is disjoint. The partitions of the species set and the reaction set always hold.

\begin{table}[h!]
\footnotesize
\begin{center}
\caption{NRP and STR Network numbers for $\mathscr{N}^*$ and the embedded networks $\mathscr{N}_i^*$}
\label{tab:NRP and STR network numbers N* }       
\begin{tabular}{ccccc|cccc}
\hline\noalign{\smallskip}
\multirow{2}{*}{Network numbers} &
      \multicolumn{4}{c}{NRP} & 
      \multicolumn{4}{c}{STR} \\
    & $\mathscr{N}^*$ & $\mathscr{N}_1^*$ & $\mathscr{N}_2^*$ & $\mathscr{N}_3^*$ &	 $\mathscr{N}^*$ & $\mathscr{N}_1^*$ & $\mathscr{N}_2^*$ & $\mathscr{N}_3^*$ \\     
\noalign{\smallskip}\hline\noalign{\smallskip}
$m$ 			& 40 & 8  & 17 & 15	& 37 & 7  & 14 & 16 \\
$m_{rev}$		& 0  & 0  &  0 &  0 & 2  & 0  & 1  & 1 \\
$n$				& 95 & 19 & 36 & 42 & 90 & 19 & 29 & 44\\
$n_r$			& 58 & 10 & 20 & 28 & 57 & 12 & 17 & 29\\
$r$				& 80 & 16 & 34 & 30 & 74 & 14 & 28 & 32\\
$l$				& 16 & 3  &  3 & 12 & 18 & 5  &  2 & 13 \\
$sl$			& 95 & 19 & 36 & 42 & 90 & 19 & 28 & 43\\
$t$				& 37 & 9  & 16 & 14 & 33 & 7  & 12 & 15\\
$s$				& 40 & 8  & 17 & 15 & 37 & 7  & 14 & 16\\
$q$				& 40 & 8  & 17 & 15 & 37 & 7  & 14 & 16\\
$\delta$		& 39 & 8  & 16 & 15 & 35 & 7  & 13 & 15\\
$\delta_p$		& 18 & 2  & 3  & 13 & 20 & 5  & 3  & 13\\
\noalign{\smallskip}\hline
\end{tabular}
\end{center}
\end{table}

The basic result on $\mathscr{N}^*$ is:
\begin{proposition}
\label{prop:independent decomposition of N*}
The embedded networks $\mathscr{N}^*_i$ of $\mathscr{N}$ constitute an independent decomposition of $\mathscr{N}^*$. 
\end{proposition}

\begin{proof}
The embedded networks $\mathscr{N}_i^*$ are subnetworks of $\mathscr{N}^*$ induced by the partition of the reaction set. Since their ranks sum up to the rank of $\mathscr{N}^*$, the decomposition is independent. 
\end{proof}

It follows from Proposition 3.3 in \cite{FLRM2017} that the deficiency of $\mathscr{N}^*$ is less than or equal to the sum of the embedded network deficiencies. It also follows from Feinberg's Decomposition Theorem, that the positive equilibria of $\mathscr{N}^*$ consists of the elements in the intersection of the positive equilibria sets of the embedded networks.

\begin{example}
The embedded networks $\mathscr{N}_1^*, \mathscr{N}_2^*$ and $\mathscr{N}_3^*$ of the NRP and STR full model constitute an independent decomposition of $\mathscr{N}^*$ which can be verified from Table \ref{tab:NRP and STR network numbers N* }. Moreover, these embedded networks form a $\mathscr{C}^*$-decomposition of $\mathscr{N}^*$ with $k(0)=k=3$. Hence, it follows from Proposition \ref{prop:k(0) for incidence-independence} that $\mathscr{N}^*$ is incidence-independent. Thus, the $\mathscr{N}_i^*$'s of the NRP and STR full model is a bi-independent $\mathscr{C}^*$-decomposition of $\mathscr{N}^*$. (This is summarized in Table \ref{tab:network decomposition of N and N* }.)
\end{example}


\section{The species subsets-induced decomposition of $\mathscr{N}$ and modularity}
\label{sec:7}

In this section, we consider $S$-systems, which, like the \textit{Mtb} NRP and STR models are derived from connected digraphs, with no independent variables and with a vertex partition defining subdigraphs.

\subsection{The species subsets-induced decomposition of $\mathscr{N}$}
An $S$-system species set partition determines not only a set of embedded networks, but also a coarsening of the $S$-system's species decomposition. Returning to the full network $\mathscr{N}$, the reaction set partition also induces a decomposition into $k$ subnetworks $\mathscr{N}_1,..., \mathscr{N}_k$. These subnetworks are not CRNs of $S$-systems since they have species which are neither dependent nor independent variables. However, they are quite similar, and, in particular $\mathscr{N}_i$ has rank $m_i$. Hence, we have the next basic fact:

\begin{proposition}
The subnetworks $\mathscr{N}_i$ constitute an independent decomposition of $\mathscr{N}$. It is incidence-independent if the species decomposition is incidence-independent.
\end{proposition}

\begin{proof}
The proof for the independence of $\mathscr{N}_i$ is similar to Proposition \ref{prop:independent decomposition of N*}. Since the subnetworks $\mathscr{N}_i$ is a coarsening of the species decomposition, the second claim follows from Proposition \ref{prop:incidence-independent coarsening}.
\end{proof}

\begin{remark}
The converse of the incidence-independent claim does not follow. It has been shown in Example \ref{example:species decomposition} that the species decomposition of NRP $\mathscr{N}$ is incidence-independent but Table \ref{tab:NRP and STR network numbers N } shows that the network decomposition of $\mathscr{N}=\mathscr{N}_1 + \mathscr{N}_2 + \mathscr{N}_3$ is incidence-independent (i.e., $n-l = \sum n_i - l_i$).
\end{remark}

We first relate the decomposition of $\mathscr{N}$ with that of $\mathscr{N}^*$. The relation of $\mathscr{N}$ and $\mathscr{N}^*$ is based on the observation that for each $i$ there is a digraph homomorphism from $\mathscr{N}_i$ to $\mathscr{N}_I^*$. A digraph homomorphism maps the vertices while preserving adjacency. For a complex in $\mathscr{N}_i$ its image is its \enquote{projection} to $\mathscr{N}_i^*$, i.e. the terms with common species are left out. 

\begin{proposition}
The map $\phi:\mathscr{N}_i \to \mathscr{N}_i^*$ is a digraph homomorphism. Each such map is surjective, and if injective, is an isomorphism of the subdigraphs.
\end{proposition}

\begin{proof}
The construction of $i^{th}$ embedded network involves removing species not in the species set $\mathscr{S}_i$ from the complexes of $\mathscr{N}_i$. Hence there is a surjective map $\phi: \mathscr{R}_i \rightarrow \mathscr{R}^*_i.$ Since any complex is a reactant or product of a reaction, one obtains a map $\phi: \mathscr{C}_i \rightarrow \mathscr{C}^*_i$ too. This map is clearly a digraph homomorphism.
\end{proof}


If the decomposition of $\mathscr{N}$ is a $\mathscr{C}^*$-decomposition, then the subdigraph homomorphism extend to a digraph homomorphism from $\mathscr{N}$ to $\mathscr{N}^*$. If, in addition, all the subdigraph maps are injective, the map $\mathscr{N}$ to $\mathscr{N}^*$ is a digraph isomorphism. 

\begin{example}
The subnetworks $\mathscr{N}_i$ is surjective digraph homomorphic to the embedded networks $\mathscr{N}_i^*$ in both the NRP and STR models. Moreover, the subnetwork $\mathscr{N}_2$ ($\mathscr{N}_1$) is injective to the embedded network $\mathscr{N}_2^*$ ($\mathscr{N}_1^*$) in NRP (STR).
\end{example}

\begin{table}[h!]
\footnotesize
\begin{center}
\caption{NRP and STR Network numbers for $\mathscr{N}$ and its subnetworks $\mathscr{N}_i$}
\label{tab:NRP and STR network numbers N }       
\begin{tabular}{ccccc|cccc}
\hline\noalign{\smallskip}
\multirow{2}{*}{Network numbers} &
      \multicolumn{4}{c}{NRP} & 
      \multicolumn{4}{c}{STR} \\
    & $\mathscr{N}$ & $\mathscr{N}_1$ & $\mathscr{N}_2$ & $\mathscr{N}_3$ &	 $\mathscr{N}$ & $\mathscr{N}_1$ & $\mathscr{N}_2$ & $\mathscr{N}_3$ \\     
\noalign{\smallskip}\hline\noalign{\smallskip}
$m$ 			& 40 & 11 & 18 & 17	& 37 & 9  & 19 & 19 \\
$m_{rev}$		& 0  & 0  &  0 &  0 & 1  & 0  & 1  & 0 \\
$n$				& 98 & 21 & 36 & 43 & 95 & 19 & 33 & 46\\
$n_r$			& 60 & 12 & 20 & 28 & 60 & 12 & 19 & 29\\
$r$				& 80 & 16 & 34 & 30 & 74 & 14 & 28 & 32\\
$l$				& 19 & 5  &  3 & 13 & 22 & 5  &  6 & 14 \\
$sl$			& 98 & 21 & 36 & 43 & 94 & 19 & 32 & 46\\
$t$				& 38 & 9  & 16 & 15 & 35 & 7  & 13 & 17\\
$s$				& 40 & 8  & 17 & 15 & 37 & 7  & 14 & 16\\
$q$				& 40 & 11 & 18 & 17 & 37 & 9  & 19 & 19\\
$\delta$		& 39 & 8  & 16 & 15 & 36 & 7  & 13 & 16\\
$\delta_p$		& 20 & 1  & 2  & 11 & 23 & 3  & 0  & 10\\
\noalign{\smallskip}\hline
\end{tabular}
\end{center}
\end{table}
 
 We now compare the network numbers of $\mathscr{N}$ and $\mathscr{N}^*$ for both the NRP and STR models. Some of the network numbers of $\mathscr{N}^*$ and $\mathscr{N}$ in Tables \ref{tab:NRP and STR network numbers N* } and \ref{tab:NRP and STR network numbers N }, respectively, are equal. In particular, their number of reactions $(r)$ are equal however, their set of reactions are not the same. The reactions in $\mathscr{N}^*$ that differ from $\mathscr{N}$ is indicated by a reaction $Ri^*$ in the Supplementary Information. Further, there are 3 more complexes in $\mathscr{N}$ than in $\mathscr{N}^*$. This is due to the construction of $\mathscr{N}^*$ which is from the union of the embedded networks $\mathscr{N}_i^*$ that involves removal of species not in the species set $\mathscr{S}_i$ from the complexes of $\mathscr{N}_i^*$ which leads to coincidence of complexes. Thus, this reduced number of complexes in $\mathscr{N}^*$ also reduces the other network numbers: $n_r, l, sl$ and $t$ in $\mathscr{N}^*$. 

Meanwhile, the network numbers of $\mathscr{N}_i^*$ and $\mathscr{N}_i$ have also similarities and differences. Their $r$, $s$ and $\delta$ are equal. There are also more species in $\mathscr{N}_i$ than in $\mathscr{N}_i^*$ because of an addition of \enquote{foreign species} (those not in the subsystem).

In terms of the network properties, both $\mathscr{N}^*$ and $\mathscr{N}$ with their subnetworks $\mathscr{N}^*_i$ and $\mathscr{N}_i$, respectively, all have the same network properties (this properties are similar to those presented in Tables \ref{tab:NRP network properties } and \ref{tab:STR network properties}).

\begin{table}[h!]
\footnotesize
\begin{center}
\caption{Description of the network decomposition of $\mathscr{N}^*$ and $\mathscr{N}$}
\label{tab:network decomposition of N and N* }       
\begin{tabular}{llccc}
\hline\noalign{\smallskip}
\multicolumn{2}{l}{Network decomposition} & Independent	& Incidence-independent	& $\mathscr{C}^*$-decomposition	\\  
\noalign{\smallskip}\hline\noalign{\smallskip}
\multirow{2}{*}{NRP}
	& $\mathscr{N}^*=\mathscr{N}^*_1+\mathscr{N}^*_2+\mathscr{N}^*_3$ & YES & YES & YES\\ 
	& $\mathscr{N}=\mathscr{N}_1+\mathscr{N}_2+\mathscr{N}_3$ & YES & YES & YES\\ 
\multirow{2}{*}{STR}
	& $\mathscr{N}^*=\mathscr{N}^*_1+\mathscr{N}^*_2+\mathscr{N}^*_3$ & YES & YES & YES\\ 
	& $\mathscr{N}=\mathscr{N}_1+\mathscr{N}_2+\mathscr{N}_3$ & YES & YES & NO\\      
\noalign{\smallskip}\hline
\end{tabular}
\end{center}
\end{table}

 The embedded networks (which corresponds to the gene subsystems) $\mathscr{N}_i^*$ in $\mathscr{N}^*$ form a $\mathscr{C}^*$-decomposition (Table \ref{tab:network decomposition of N and N* }). This means that their only interactions are \enquote{outside} the system as represented by the zero complex. The subnetworks $\mathscr{N}_i$ in $\mathscr{N}$ on the other hand, have reactions where foreign species are involved and it even have complexes in common other than the zero complex. In fact, the non-zero complex \enquote{$Y10$} is common in the STR $\mathscr{N}_2$ and $\mathscr{N}_3$ which makes the STR $\mathscr{N}=\mathscr{N}_1 + \mathscr{N}_2 + \mathscr{N}_3$ not a $\mathscr{C}^*$-decomposition (Table \ref{tab:network decomposition of N and N* }). 
 
\subsection{Modularity of digraph divisions and modularity of CRN decompositions}

Biological systems often display an organization into functional modules and hence, it is important that models of such systems capture characteristics of the modular structure. In this Section, we discuss the modularity concepts available for digraph divisions, but provide evidence that these need to be modified or extended for CRN decompositions.

Modularity was initially introduced by Newman and Girvan for the case of undirected networks while Arenas proposed an extension of this for directed networks \cite{LISC2011}. Their extension is based on the observation that the existence of a directed edge $(i,j)$ between nodes $i$ and $j$, depends on the out-degree and in-degree of nodes $i$ and $j$ respectively. The modularity  for directed network, denoted by $Q$, is expressed as:

$$Q=\frac{1}{m} \sum_{i,j} [ A_{ij}- \frac{k_i^{out} k_j^{in}}{m}] \delta(c_i,c_j) $$
where $m$ is the total number of arcs in the network, $A_{ij}$ is the number of arcs from $i$ to $j$, $k_i^{out}$ and $k_j^{in}$ are the outdegree and indegree of the nodes $i$ and $j$, respectively and $\delta(c_i,c_j)=1$ (i.e., if nodes $i$ and $j$ belong on the same module) and $0$ otherwise.

In terms of the structure or graph of a network, modularity is designed to measure the strength of division of a network into modules (clusters or communities). Good divisions, which have high modularity values, are those with dense edge connections between the vertices within a module but sparse connections between vertices in different modules \cite{LISC2011}. 

A systems biologist describes modules from a graph-theoretical point of view as a group of nodes that are more strongly intraconnected than interconnected while a geneticist might consider a set of co-expressed or co-regulated genes a module \cite{LOJD2011}. 

When we compute the modularities $Q$ of the divisions of the digraphs in Figure \ref{fig:NRPSTR} and the reaction graph of the CRN decompositions of $\mathscr{N}$ (which is also a digraph) for NRP and STR, we obtain the following surprising results:

\begin{table}[h!]
\begin{center}
\caption{Modularity $Q$ of the divisions of the digraph and the reaction graph of $\mathscr{N}$ for NRP and STR}
\label{tab:Q of the digraph and CRN}
\begin{tabular}{ccc}
\hline\noalign{\smallskip}
\textit{Mtb} $S$-system & $Q$ (digraph) & $Q$ (reaction graph)\\
\noalign{\smallskip}\hline\noalign{\smallskip}
NRP & $0.4405$ & $0.3217$ \\
STR & $0.4472$ & $0.1989$ \\
\noalign{\smallskip}\hline
\end{tabular}
\end{center}
\end{table}

Paradoxically, the digraph model indicates that STR is more modular than NRP, while the same measure applied to the reaction graph of the CRN indicates the opposite. Our conclusion from this computation is that because a chemical reaction network has a richer structure than just a digraph, one needs to modify or expand the concept of modularity for a CRN to include aspects of its stoichiometric structure. In the following, we introduce an initial concept which is admittedly specific for $S$-system CRNs but might lead in the right direction.

Now, if two subnetworks (i.e. subdigraphs) are in different connected components, then they are \enquote{physically} isolated, and hence, the question of modularity is trivial. We hence assume in the following that the subdigraphs in question are in the same connected component. 

Any arc in the digraph (biochemical map) whose source and target vertices lie in different subdigraphs lead to the occurrence of common species between the two subnetworks of $\mathscr{N}$ in the $S$-system CRN's. Since by assumption, the digraph under consideration is connected, for each subdigraph, there is at least one such arc and hence at least one common species with another subnetwork. Based on this, we introduce the concept of species coupling level of a subnetwork:

\begin{definition}
The \textbf{species coupling level $c_S(N')$} of a subnetwork $N'$ in $N$ is the ratio of the number of occurrences of common species in the reactant and product complexes of the input and output reactions, respectively and the number of occurrences of non-common species in the reactant complexes of the input reactions.
\end{definition} 


With respect to the digraph of the $S$-system, the numerator counts the number of arcs between vertices from different subdigraphs coming into the subdigraph while the denominator counts the number of arcs between vertices within the subdigraph.

\begin{remark}
The species coupling level of a network is equal to the sum of the species coupling levels of its subnetworks. 
\end{remark}

Table \ref{tab:Q and species coupling level} compares the modularity of the digraphs with the species coupling levels for NRP and STR.

\begin{table}[h!]
\begin{center}
\caption{Modularity $Q$ of the digraphs and the species coupling levels $\sum c_s (\mathscr{N}_i)$ for NRP and STR}
\label{tab:Q and species coupling level}
\begin{tabular}{ccc}
\hline\noalign{\smallskip}
\textit{Mtb} $S$-system & $Q$ (digraph) & $\sum c_s (\mathscr{N}_i)$\\
\noalign{\smallskip}\hline\noalign{\smallskip}
NRP & $0.4405$ & $0.6236$ \\
STR & $0.4472$ & $0.5289$ \\
\noalign{\smallskip}\hline
\end{tabular}
\end{center}
\end{table}

Since modularity is inversely related to species coupling, we see that the values of the latter for NRP and STR are qualitatively consistent with the modularity values for the corresponding digraphs. In our view, this indicates that stoichiometric level information is useful for modularity considerations.

The NRP and STR decompositions contain only mono-species common complexes. We counted the single non-zero one in the species-level coupling calculation and considered the common zero complex as \enquote{outside of the system} studied. We would view the occurrence of common multi-species complexes in other examples as further indicators on the stoichiometric level of lower modularity. We hope to explore these and related concepts in a more general context of developing an appropriate modularity concept for chemical reaction network decompositions.


\section{Conclusions and outlook}
Motivated by Magombedze and Mulder’s \cite{MAMU2013} approach in representing and analyzing the gene regulatory based system model of \textit{Mtb} in modular form, where the system partitions the vertices into subsets that form corresponding subsystems, we studied these subsystems as embedded networks. In the embedded CRN representation of S-system, it is defined as a species subset and a reaction subset induced by the set of digraph vertices of the subsystem.

In this study, we have
\begin{enumerate}

\item corrected Proposition 10 of \cite{AJMM2015} and provided a new proof using decomposition theory. 

\item presented the discordance of $S$-system CRN with at least two dependent species. 

\item developed the concept of $\mathscr{C}$- and $\mathscr{C}^*$-decompositions including their structure theorems in terms of linkage classes and defined a new $S$-system CRN $\mathscr{N}^*$. The union of the embedded networks $\mathscr{N}_i^*$ of $\mathscr{N}$  construct the new $S$-system CRN $\mathscr{N}^*$. Note that these $\mathscr{N}_i^*$s do not partition the reaction set in $\mathscr{N}$ but constitute an independent $\mathscr{C}^*$-decomposition of $\mathscr{N}^*$.

\item shown that the species subsets-induced decomposition of $\mathscr{N} = \sum \mathscr{N}_i$ induces surjective digraph homomorphisms between $\mathscr{N}_i$ and $\mathscr{N}_i^*$.

\item illustrated that the modularity of the decomposition of the $S$-system $\mathscr{N}$ is not consistent with the modularity of the original digraph model of the gene regulatory system. In order for the modularity concept from digraph to capture the stoichiometric structure of CRN, we have introduced the concept of species coupling level for the CRN decompositions. 


\end{enumerate}

Biological systems are complex networks that exhibit the orchestrated interplay of a large array of components \cite{KIM2003}. It is common practice for modellers to decompose the complex system into subsystems. It is believed that studying the dynamics and functionality of the subsystems would facilitate understanding of the system as a whole. Henceforth, discovering and analyzing such subsystems are crucial in gaining better understanding of the complex systems \cite{KIM2003}. There are still a lot to explore and study regarding the topological/network properties with respect to decomposition/separability of subsystems.


\vspace{10pt}

\noindent \textbf{Acknowledgments}. HFF acknowledges the support of the Commission on Higher Education (CHED), Philippines for the CHED-SEGS Scholarship Grant. ARL held research fellowship from De La Salle University and would like to acknowledge the support of De La Salle University’s Research Coordination Office.

\newpage
\appendix

\section{Fundamentals of chemical reaction networks and kinetic systems}
\label{app:fundametals of CRNT}
We recall the necessary concepts of chemical reaction networks and the mathematical notation used throughout the paper adopted from the papers \cite{AJMM2015}, \cite{FEIN1987} and \cite{FLRM2017}. 

We begin with the definition of a chemical reaction network.
\begin{definition}
A \textbf{chemical reaction network} is a triple $\mathscr{N}=(\mathscr{S},\mathscr{C},\mathscr{R})$ of three non-empty finite sets:
\begin{enumerate}
\item A set \textbf{species} $\mathscr{S}$,
\item A set $\mathscr{C}$ of \textbf{complexes}, which are non-negative integer linear combinations of the species, and
\item A set $\mathscr{R} \subseteq \mathscr{C} \times \mathscr{C}$ of \textbf{reactions} such that
\begin{itemize}
\item $(i,i) \notin \mathscr{R}$ for all $i \in \mathscr{C}$, and
\item For each $i \in \mathscr{C}$, there exists a $j \in \mathscr{C}$ such that $(i,j) \in \mathscr{R}$ or $(j,i) \in \mathscr{R}.$
\end{itemize}
\end{enumerate}

\end{definition}

We denote with $m$ the number of species, $n$ the number of complexes and $r$ the number of reactions in a CRN.

\begin{definition}
A complex is called \textbf{monospecies} if it consists of only one species, i.e. of the form $kX_i$, $k$ a non-negative integer and $X_i$ a species. It is called \textbf{monomolecular} if $k=1$, and is identified
with the \textbf{zero complex} for $k=0$.
\end{definition}
 A zero complex represents the \enquote{outside} of the system studied, from which chemicals can flow into the system at a constant rate and to which they can flow out at a linear rate (proportional to the concentration of the species). In biological systems, the \enquote{outside} also stands for the degradation of a species. 

A chemical reaction network $(\mathscr{S}, \mathscr{C}, \mathscr{R})$ gives rise to a digraph with complexes as vertices and reactions as arcs. However, the digraph determines the triple uniquely only if an additional property is considered in the definition: $\mathscr{S}=\bigcup$ supp $i$ for $i \in \mathscr{C},$ i.e., each species appears in at least one complex. With this additional property, a CRN can be equivalently defined as follows.

\begin{definition}
A \textbf{chemical reaction network} is a digraph $(\mathscr{C}, \mathscr{R})$ where each vertex has positive degree and stoichiometry, i.e., there is a finite set $\mathscr{S}$ (whose elements are called \textbf{species}) such that $\mathscr{C}$ is a subset of $\mathbb{R}^{\mathscr{S}}_{\geq}.$ Each vertex is called a \textbf{complex} and its coordinates in $\mathbb{R}^{\mathscr{S}}_{\geq}$ are called \textbf{stoichiometric coefficients}. The arcs are called \textbf{reactions}.
\end{definition}

Two useful maps are associated with each reaction:

\begin{definition}
The \textbf{reactant map} $\rho: \mathscr{R} \rightarrow \mathscr{C}$ maps a reaction to its reactant complex while the \textbf{product map} $\pi: \mathscr{R} \rightarrow \mathscr{C}$ maps it to its product complex. We denote $|\rho (\mathscr{R})|$ with $n_r$, i.e., the number of reactant complexes. 
\end{definition}

Connectivity concepts in Digraph Theory apply to CRNs, but have slightly differing names. A connected component is traditionally called a \textbf{linkage class}, denoted by $\mathscr{L}$, in CRNT. A subset of a linkage class where any two elements are connecteed by a directed path in each direction is known as a \textbf{strong linkage class}. If there is no reaction from a complex in the strong linkage class to a complex outside the same strong linkage class, then we have a \textbf{terminal strong linkage class}. We denote the number of linkage classes with $l$, that of the strong linkage classes with $sl$ and that of terminal strong linkage classes with $t$. Clearly, $sl \geq t \geq l.$

Many features of CRNs can be examined by working in terms of finite dimensional spaces $\mathbb{R}^{\mathscr{S}}, \mathbb{R}^{\mathscr{C}}, \mathbb{R}^{\mathscr{R}},$ which are referred to as species space, complex space and reaction space, respectively. We can view a complex $j \in \mathscr{C}$ as a vector in $\mathbb{R}^{\mathscr{C}}$ (called \textit{complex vector}) by writing $j = \sum _{s \in \mathscr{S}} j_s s,$ where $j_s$ is the stoichiometric coefficient of species $s$.

\begin{definition}
The \textbf{reaction vectors} of a CRN $(\mathscr{S}, \mathscr{C}, \mathscr{R})$ are the members of the set $\{j-i \in \mathbb{R}^{\mathscr{S}} | (i,j) \in \mathscr{R}\}.$ The \textbf{stoichiometric subspace} $S$ of the CRN is the linear subspace of $\mathbb{R}^{\mathscr{S}}$ defined by 
	$$S: span \{j-i \in \mathbb{R}^{\mathscr{S}} | (i,j) \in \mathscr{R}\}.$$
The \textbf{rank} of the CRN, $s$, is defined as $s=dim S.$
\end{definition}

\begin{definition}
The \textbf{incidence map} $I_a: \mathbb{R}^{\mathscr{R}} \rightarrow \mathbb{R}^{\mathscr{C}}$ is defined as follows. For $f: \mathscr{R} \rightarrow \mathbb{R}$, then $I_a(f)(v) = - f(a)$ and $f(a)$ if $v = \rho(a)$ and $v = \pi(a)$, respectively, and are $0$ otherwise.
\end{definition}

\noindent Equivalently, it maps the basis vector $\omega_a$ to  $\omega_{v'} -  \omega_v$ if $a: v \rightarrow v'$.\\

It is clearly a linear map, and its matrix representation (with respect to the standard bases $\omega_a$, $\omega_{v}$) is called the \textbf{incidence matrix}, which can be described as 
\begin{center}
\[
 (I_a)_{i,j} = 
  \begin{cases} 
   -1 & \text{if } \rho(a_j) = v_i, \\
   1       & \text{if } \pi(a_j) = v_i,\\
   0		& \text{otherwise}.
  \end{cases}
\]
\end{center}
\noindent Note that in most digraph theory books, the incidence matrix is set as $-I_a$.\\
An important result of digraph theory regarding the incidence matrix is the following:

\begin{proposition}
\label{prop:IncidenceMatrix02}
Let $I$ be the incidence matrix of the directed graph $D = (V, E)$. Then rank $I = n -l$, where $l$ is the number of connected components of $D$.
\end{proposition}

A non-negative integer, called the deficiency, can be associated to each CRN. This number has been the center of many studies in CRNT due to its relevance in the dynamic behavior of the system.

\begin{definition}
The \textbf{deficiency} of a CRN is the integer $\delta = n-l-s.$ 
\end{definition}

\begin{definition}
The \textbf{reactant subspace} $R$ is the linear space in $\mathbb{R}^\mathscr{S}$ generated by the reactant complexes. Its dimension, $dim R$ denoted by $q$, is called the \textbf{reactant rank} of the network. Meanwhile, the \textbf{reactant deficiency} $\delta_p$ is the difference between the number of reactant complexes and the reactant rank, i.e., $\delta_p = n_r -q.$
\end{definition}

We now introduce the fundamentals of chemical kinetic systems. We begin with the general definitions of kinetics from Feliu and Wiuf \cite{FW2012}:

\begin{definition}
A \textbf{kinetics} for a CRN $(\mathscr{S}, \mathscr{C}, \mathscr{R})$ is an assignment of a rate function $K_j: \Omega_K \rightarrow \mathbb{R}_\geq$ to each reaction $r_j \in \mathscr{R}$, where $\Omega_K$ is a set such that $\mathbb{R}^\mathscr{S}_> \subseteq \Omega_K \subseteq \mathbb{R}^\mathscr{S}_\geq$, $c\wedge d \in \Omega_K$ whenever $c,d \in \Omega_K,$ and 
$$K_j(c) \geq 0, \quad \forall c \in \Omega_K.$$
A kinetics for a network $\mathscr{N}$ is denoted by $\displaystyle{K=(K_1,K_2,...,K_r):\Omega_K \to {\mathbb{R}}^{\mathscr{R}}_{\geq}}$. The pair $(\mathscr{N}, K)$ is called the \textbf{chemical kinetic system} (CKS).
\end{definition}

\noindent
In the definition, $c \wedge d$ is the bivector of $c$ and $d$ in the exterior algebra of $\mathbb{R}^\mathscr{S}.$ We add the definition relevant to our context:

\begin{definition}
A chemical kinetics is a kinetics $K$ satisfying the positivity condition: for each reaction $j:y\rightarrow y', K_j(c)>0$ iff $\supp y\subset\supp c$.
\end{definition}

Once a kinetics is associated with a CRN, we can determine the rate at which the concentration of each species evolves at composition $c$.


Power-law kinetics is defined by an $ r x m$ matrix $F =[F_{ij}],$ called the \textbf{kinetic order matrix}, and vector $k \in \mathbb{R}^\mathscr{R}$, called the \textbf{rate vector}. In power-law formalism, the kinetic orders of the species concentrations are real numbers. 

\begin{definition}
A kinetics $K: \mathbb{R}^\mathscr{R}_> \rightarrow \mathbb{R}^\mathscr{R}$ is a \textbf{power-law kinetics} (PLK) if
$$K_i (x) = k_ix^{F_i} \quad \forall i = 1, ... , r$$
with $k_i \in \mathbb{R}_>$ and $ F_{ij} \in \mathbb{R}.$
\end{definition}

\section{Network decomposition}
\label{app:network decomposition}
The following diagram illustrates several decompositions introduced in this paper showing its relationship to Feinberg's decomposition theory.

\begin{figure}[h!]
\begin{center}
\includegraphics[width=1.1\textwidth]{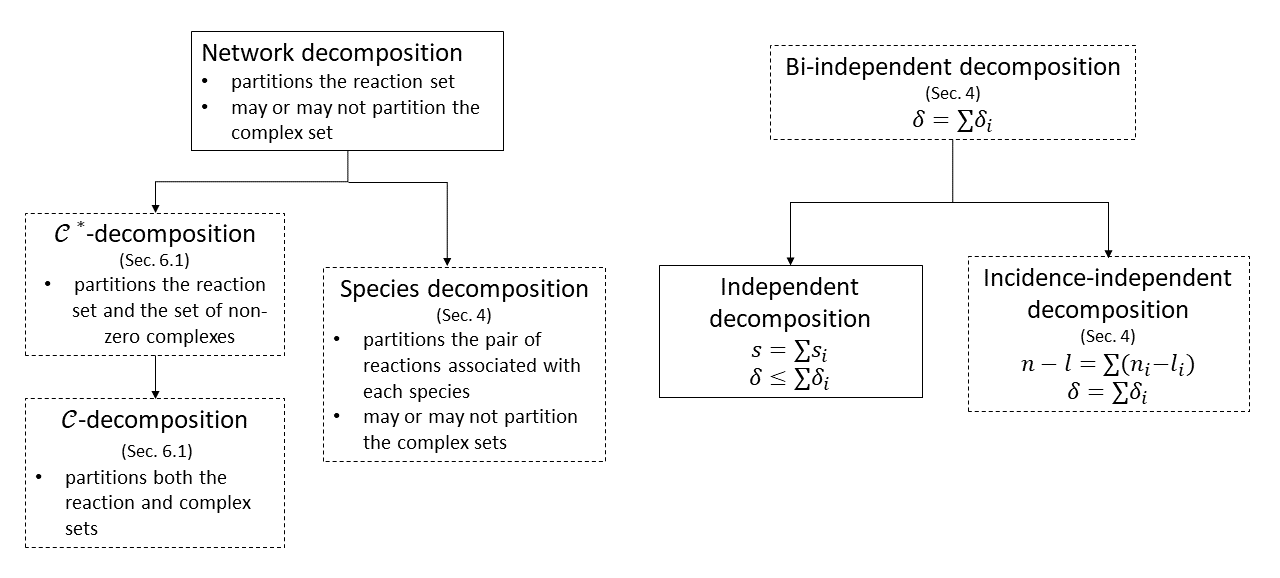} 
\caption{Relationship and properties of Feinberg's decomposition theory (in solid-lined boxes) to the new terms (in dash-lined boxes) that emerged from this paper.}
\label{Network decomposition}
\end{center}
\end{figure}

\end{document}